%% file: clique-listing.tex
\documentclass{article}

\usepackage[utf8]   {inputenc}

\usepackage[margin=3.4cm]{geometry}
\usepackage{amsmath} 
\usepackage{pgfplots}
\usepackage{amsthm}
\usepackage{amssymb} 
\usepackage{complexity}
\usepackage{mathrsfs}
\usepackage[numbers,sectionbib]{natbib} 
\usepackage{paralist}
\usepackage{tikz}
\usetikzlibrary{shapes}
\usepackage[backgroundcolor=white,linecolor=black]{todonotes}

\usepackage         {xparse}

\input{notation}

\usepackage         {cleveref} 
\Crefname{observation}{Observation}{Observations}
\Crefname{hypothese}{Assumption}{Assumptions}

\title{An Improved Upper Bound on 
Maximal Clique Listing via Rectangular Fast Matrix Multiplication
\footnote{This work was supported by the \emph{Department of Computer Science, University of Verona, Italy}, under Ph.D. grant 
``Computational Mathematics and Biology`` on a co-tutelle agreement with \emph{LIGM, Universit\'e Paris-Est in Marne-la-Vall\'ee, Paris, France}.}} 

\author{Carlo Comin\footnote{\emph{Department of Mathematics, University of Trento, Italy}. (carlo.comin@unitn.it)} 
 \and Romeo Rizzi\footnote{\emph{Department of Computer Science, University of Verona, Italy}. (romeo.rizzi@univr.it)}} 

\date{} 

\begin{document} 

\maketitle

\begin{abstract} 
The first output-sensitive algorithm for the Maximal Clique Listing problem was given by Tsukiyama \etal in 1977. 
As any algorithm falling within the Reverse Search paradigm, it performs a DFS visit of a directed tree (the RS-tree) 
having the objects to be listed (\ie maximal cliques) as its nodes. 
In a recursive implementation, the RS-tree corresponds to the recursion tree of the algorithm. 
The time delay is given by the cost of generating the next child of a node, and Tsukiyama showed it is $O(mn)$. 
In 2004, Makino and Uno sharpened the time delay to $O(n^{\omega})$ 
by generating all the children of a node in one single shot, which is performed by computing a \emph{square} fast matrix multiplication.
In this paper, we further improve the asymptotics for the exploration of the same 
RS-tree by grouping the offsprings' computation even further. Our idea is to rely on \emph{rectangular} fast matrix multiplication 
in order to compute all the children of $n^2$ nodes in one single shot. According to the current upper bounds on square and rectangular 
fast matrix multiplication, with this the time delay improves from $O(n^{2.3728639})$ to $O(n^{2.093362})$.
\end{abstract} 

\emph{Keywords:} Maximal Clique Listing, Rectangular Fast Matrix Multiplication, 
Output Sensitive Algorithm, Polynomial Time Delay, Reverse Search Enumeration, Backtracking.

\input{clique-listing-Sect1-Introduction}
\input{clique-listing-Sect2-BandN}
\input{clique-listing-Sect3-RS-Tree}

\input{clique-listing-Sect4-Reduction}
\input{clique-listing-Sect5-Backtracking}

\input{clique-listing-Sect6-Algorithm}

\section{Conclusion}\label{sect:conclusion}
In this paper, we improved the asymptotics for the exploration of the 
RS-tree associated to {\mainproblem} by grouping the offsprings' computation in a novel manner. 
In summary, our idea was to rely on \emph{rectangular} fast matrix multiplication 
in order to compute all the children of $n^2$ maximal cliques in one single shot.
The major open question, on this way, is that to understand whether or not the {\mainproblem} problem admits  
$O(n^{2+o(1)})$ time delay algorithms that meanwhile maintain both the bootstrapping time and the working space polynomial in $n$. 

\thispagestyle{plain}  
\pagestyle{plain} 


\bibliographystyle{mynatstyle}
\bibliography{biblio}

\label{BibliographyPage}
\clearpage
\appendix\label{app:time_complexity}
\section{\appendixname~A: Analysis of the Time Complexity (Extended Version)}
\label{app:time_complexity}
\input{appendix-time_complexity}

\section{\appendixname~B: A Reduction from $\I^{P}$ to {\sc QSFI}}
\label{app:bitwise_and}
\input{appendix-bitwise_and}

\end{document}

%% file: notation.tex
\newcommand{\ie}{i.e.,\ }
\newcommand{\eg}{e.g.,\ }

\newcommand{\closure}{\texttt{lc}}
\newcommand{\parent}{\P}

\newcommand{\mainproblem}{{\sc MCL}}

\tikzstyle myBG=[line width=3pt,opacity=1]
 
\newcommand{\drawLinewithBG}[2]
{
  \draw[white,myBG]  (#1) -- (#2);
  \draw[black,very thick] (#1) -- (#2);
}
\newcommand{\drawPolarLinewithBG}[2]
{
  \draw[white,myBG]  (#1) -- (#2);
  \draw[black,very thick] (#1) -- (#2);
}

\newcommand{\etal}{\textit{et al.}\xspace}
\newcommand{\wrt}{w.r.t.\ }

\newcommand{\suchthat}{\;\ifnum\currentgrouptype=16 \middle\fi|\;}

\newcommand{\RR}{\mathbb{R}\xspace}
\newcommand{\Q}{\mathbb{Q}\xspace}
\newcommand{\N}{\mathbb{N}\xspace}

\def\C{{\cal C}}
\def\I{{\cal I}}

\def\B{{\cal B}}

\def\K{{\cal K}}

\def\T{{\cal T}}
\def\P{{\cal P}}

\newcommand{\figref}[1]{Fig.~\ref{#1}}

\newtheorem{Thm}{Theorem}

\newtheorem{Lem}{Lemma}
\newtheorem{Prop}{Proposition}

\makeatletter
\providecommand*{\cupdot}{%
  \mathbin{%
    \mathpalette\@cupdot{}%
  }%
}
\newcommand*{\@cupdot}[2]{%
  \ooalign{%
    $\m@th#1\cup$\cr
    \sbox0{$#1\cup$}%
    \dimen@=\ht0 %
    \sbox0{$\m@th#1\cdot$}%
    \advance\dimen@ by -\ht0 %
    \dimen@=.5\dimen@
    \hidewidth\raise\dimen@\box0\hidewidth
  }%
}

\newcommand\restr[2]{{
  \left.\kern-\nulldelimiterspace 
  #1 
  \vphantom{\big|} 
  \right|_{#2} 
}}

\usepackage[ruled,vlined,linesnumbered]{algorithm2e}
\usepackage{subfig}
\SetCommentSty{textsf}
\SetKwRepeat{DoWhile}{do}{while}
\SetAlFnt{\small} 
\makeatletter
\newcommand{\removelatexerror}{\let\@latex@error\@gobble}
\makeatother

\let\oldnl\nl
\newcommand{\nonl}{\renewcommand{\nl}{\let\nl\oldnl}}

\usepackage{tikz-cd}

\tikzset{
  shorten/.style={/tikz/shorten >={#1},/tikz/shorten <={#1}}}

\usepackage{tikz-qtree}
\usetikzlibrary{calc,positioning,fit}
\usetikzlibrary{shapes,shapes.multipart,shapes.arrows}
\usetikzlibrary{decorations,decorations.pathmorphing,decorations.pathreplacing,decorations.markings,decorations.shapes}
\usetikzlibrary{arrows}
\usetikzlibrary{fit,backgrounds}
\usetikzlibrary{plotmarks}
\tikzstyle{node}=[circle,draw,inner sep=2pt,transform shape,minimum size=1.75em]

%% file: clique-listing-Sect1-Introduction.tex
\section{Introduction}\label{sect:introduction}
In an undirected graph $G$, a \emph{clique} is any subset $K$ 
of the vertex set such that any two vertices in $K$ are adjacent. 
A clique is \emph{maximal} when it is not a subset of any larger clique. 
This paper addresses the problem of generating all the maximal cliques of a given graph, 
namely \emph{Maximal Clique Listing} (\mainproblem).
Maximal cliques are fundamental graph objects, so the {\mainproblem} problem  
may be regarded as one of the central problems in the field of graph enumeration, 
and indeed it attracted a considerable attention also in the 
past~\cite{TsukiyamaIAS77, Chiba85, Johnson88, MU04}.  
The problem has not only theoretical interest in computational complexity, 
but it possesses several consolidated applications as well, \eg in bioinformatics, 
clustering, computational linguistics and 
data-mining~\cite{MU04, Eppstein10, Eppstein11}.

As shown by Moon and Moser in 1965, any graph on $n$ vertices contains at most $3^{n/3}$ 
maximal cliques~\cite{MoonMoser65}. 
It is therefore particularly interesting to focus on \emph{polynomial time delay} 
algorithms for generating all of them without repetitions. 
An {\mainproblem} algorithm has $O(f(n))$ \emph{time delay}  
whenever the time spent between the outputting of any two consecutive maximal cliques is $O(f(n))$;
for this, the procedure is allowed to undertake a polynomial time pre-processing phase, if needed. 
 
Both in the past and more recently, a considerable number of algorithms have been presented 
and evaluated (experimentally or theoretically) for {\mainproblem}.
Tsukiyama, \etal~\cite{TsukiyamaIAS77} first proposed in 1977 a polynomial time delay solution for generating 
all maximal independent sets (thus, by complementarity, all maximal cliques) in a given graph $G=(V,E)$. 
Their procedure works with $O(n+m)$ space and $O(mn)$ time delay. Here, $m=|E|$ and $n=|V|$.
In 1985, Chiba and Nishizeki~\cite{Chiba85} reduced the time delay to $O(\gamma(G)m)$, 
where $\gamma(G)$ is the arboricity of $G$ and $m/(n-1)\leq \gamma(G)\leq m^{1/2}$. 
Johnson, Yannakakis and Papadimitriou~\cite{Johnson88} proposed in 1988 
an algorithm for enumerating all the maximal cliques in the lexicographical order. 
Their procedure runs with $O(mn)$ time delay, but it also uses $O(nN)$ space 
(where $N$ denotes the total number of maximal cliques of $G$).
A summary of previous and present results is offered in Table~\ref{Table:Algorithms}.

\begin{table}[!htb]
\caption{Time and Space Complexity of the main Output-Sensitive Algorithms for {\mainproblem}.}
\label{Table:Algorithms}
\centering
\bgroup
\def\arraystretch{1.4}
\begin{tabular}{| c   c  c  c | }
\hline
\vtop{\hbox{\strut }\hbox{\strut Algorithm }\hbox{\strut }}& \vtop{\hbox{\strut }\hbox{\strut Time to First $x$ }\hbox{\strut }} & \vtop{\hbox{\strut }\hbox{\strut Time Delay }\hbox{\strut }}  
  & \vtop{\hbox{\strut }\hbox{\strut Work Space }\hbox{\strut }} \\
\hline
	Algo.~\ref{ALGO:strict} 
	& 
	$O(n^{\omega+3} + x n^{2\omega(1,1,1/2)-2})$ 
	& 
	$O(n^{2\omega(1,1,1/2)-2})$ 
 &   $O(n^{4.2796})$  \\ 
	Algo.~\ref{ALGO:main} 
	& 
	$O(n^{\omega+3} + x n^{2\omega(1,1,1/2)-2})$ 
& 
	$O(n^{\omega+2})$ 
&  $O(n^4)$ \\ 
	MU04~\cite{MU04} 
  & 
	$O(xn^{\omega})$ 
	& 
	$O(n^{\omega})$ 
&    $O(n^2)$ \\
CN85~\cite{Chiba85} &  $O(x \alpha(G)m)$ & $O(\alpha(G)m)$ &   $O(m+n)$   \\
TIAS77~\cite{TsukiyamaIAS77} &  $O(xmn)$ & $O(mn)$ & $O(m+n)$ \\
JYP88~\cite{Johnson88} & $O(xmn)$  &  $O(mn)$ &  $O(mnN)$ \\
\hline
\end{tabular}
\egroup
\end{table}


Both the algorithm of Tsukiyama, \etal and that of Johnson, \etal can be placed within the framework  
of the \emph{Reverse Search Enumeration} (\emph{RSE}), a technique which was first introduced by Avis and Fukuda 
in the context of efficient enumeration of vertices of polyhedra and arrangements of hyperplanes~\cite{AvisFukuda93}.
Very briefly, the RSE is a technique for listing combinatorial objects by \emph{reversing} a given optimization objective function $f$.
Let $G=(V,E)$ be a connected graph whose nodes are precisely the objects to be listed. 
Suppose we have some objective function $f:V\rightarrow\N$ to be maximized over all nodes of $G$. 
Also, assume we are given a local search algorithm on $G$ that is a deterministic procedure to move 
from any node to some neighboring node which is larger with respect to $f$, until there exists no better neighbor. 
The algorithm is finite if for any starting node it terminates within a finite number of steps. 
We may consider the digraph $\T_G$ with the same node set as $G$ and in which the edges are all the ordered pairs $(x,x')$ of 
consecutive nodes $x$ and $x'$ generated by the same local search algorithm. Assuming that there is only one local optimal node $x^*$, then 
$\T_G$ is a single directed tree spanning all the nodes of $G$ and having $x^*$ as its only sink and root. 
In this manner, if we trace $\T_G$ from $x^*$ backwards, say with a Depth-First Search, 
we can enumerate all nodes of $G$, i.e., all combinatorial objects. The major operation involved is tracing each edge against its orientation, 
which corresponds to \emph{reversing} the local search optimization algorithm in order to compute 
a parent-child relation that fully describes $\T_G$; notice that, in this case, the minor work of backtracking is 
simply that of performing a single local search step itself. Whence, the key ingredient of any RSE is the computation of 
the parent-child relation in an efficient way. If the height of $\T_G$ is at most $n$, 
then the memory consumed throughout the listing process is polynomial in $n$.

Indeed, the algorithm of Tsukiyama,~\etal performs a DFS visit of a directed tree -- namely, the RS-tree -- having the objects to be listed (\ie maximal cliques) as its nodes. 
In a recursive implementation, the RS-tree corresponds to the recursion tree of the algorithm. Tsukiyama, \etal showed in \cite{TsukiyamaIAS77} that the time delay of their procedure is $O(mn)$. 
In 2004, Makino and Uno sharpened the time delay of {\mainproblem} to $O(n^{\omega})$, by generating all the 
children of a node in one single shot which is performed by computing a \emph{square} fast matrix multiplication~\cite{MU04}. 
In particular, the procedure of Makino and Uno runs with $O(M(n))=O(n^{\omega})$ time delay and works with $O(n^2)$ space,  
where $M(n)=O(n^{\omega})$ denotes the minimum number of arithmetic operations needed to multiply two $n\times n$ \emph{square} matrices. 
The best upper bound on $\omega$ which is currently known was shown by Le~Gall~\cite{Gall14} in 2014, it is $\omega \leq 2.3728639$. 
By these results, the algorithm of Makino and Uno runs with $O(n^{2.3728639})$ time delay. 
To the best of our knowledge, this is the tightest upper bound on the time 
delay complexity of {\mainproblem} which is currently known in the literature. 

\smallskip
\textbf{Contribution.} 
In this work we improve the tightest known upper bound on the time delay complexity of {\mainproblem}.  
In particular, we show that the parent-child relation of the corresponding RS-Tree   
admits an asymptotically faster (with respect to that devised by Makino and Uno~\cite{MU04}) computing procedure
that works by grouping the offsprings' computation even further than in~\cite{MU04}.
Briefly, our procedure works by grouping together multiple children generation problems into batches 
of $n^2$ problems (where each single problem consists into the task of computing all children of a given maximal clique) 
and then by reducing the job of solving a whole batch of $n^2$ problems, in one single shot,   
to that of multiplying two \emph{rectangular} matrices. We remark that, in so doing, 
this work proposes a novel representation for the basic task of generating the children nodes for {\mainproblem}. 
This conceptual shift is the essential lever which allows for a suitable adoption of 
rectangular matrix multiplication methods in {\mainproblem}. In this way, 
we prove a sharpened upper bound on the time delay of {\mainproblem}, 
improving it from $O(n^{\omega})=O(n^{2.3728639})$ to $O(n^{2\omega(1,1,1/2)-2})=O(n^{2.093362})$; 
here, $O(n^{2\omega(1,1,1/2)})$ denotes the minimum number of arithmetic operations needed 
to perform any $n^2\times n$ by $n\times n^2$ matrix product, 
and it is the standard\footnote{
Following the notation as in~\cite{Gall12}, for any $k\in\Q$ such that $k>0$, 
let $C(n, \lfloor n^k \rfloor, n)$ be the minimum number of \emph{arithmetic} 
operations needed to multiply an $n\times \lfloor n^k \rfloor$ matrix by an $\lfloor n^k \rfloor\times n$ one. 
The corresponding complexity exponent is defined as follows:
\[
\omega(1,k,1) = \inf\{ \tau\in\RR \mid C(n,\lfloor n^k \rfloor, n)=O(n^\tau)\},\;\;\;\; \text{for every $(0, +\infty)\cap\Q$.} 
\]
Notice that $\omega(1,1,1)=\omega$ is the complexity exponent of the $n\times n$ square case.
As for $ \lfloor n^i\rfloor \times \lfloor n^k \rfloor$ by $\lfloor n^k \rfloor\times \lfloor n^j\rfloor$ matrix products, 
the corresponding complexity exponent is $\omega(i,k,j) = 
\inf\{ \tau\in\RR \mid C( \lfloor n^i \rfloor, \lfloor n^k \rfloor,\lfloor n^j \rfloor)=O(n^\tau)\}$. 
} 
notation for expressing significant bounds on rectangular matrix multiplication. Our main results are summarized below.

\begin{Thm}\label{thm:firstx}
There is a procedure (Algorithm~\ref{ALGO:main}) for listing all the 
maximal cliques of any given $n$-vertex graph $G=(V,E)$, without repetitions, 
and in such a way that for every $x\in\N$ the first $x$ maximal cliques are outputted within the following time bound:
\[\tau_{\texttt{first\_}x}= O\big(n^{\omega+3}+xn^{2\omega(1,1,1/2)-2}\big) = O\big(n^{5.3728639} + xn^{2.093362}\big).\] 
For this, the procedure employs $O(n^4)$ space.
\end{Thm}

\begin{Thm}\label{thm:strict}
There is a procedure (Algorithm~\ref{ALGO:strict}) for listing all the 
maximal cliques of any given $n$-vertex graph $G=(V,E)$, without repetitions, and with the following time delay:
\[\tau_{\texttt{delay}}= O\big(n^{2\omega(1,1,1/2)-2}\big) = O\big(n^{2.093362}\big),\] 
which is a worst-case upper bound on the time spent between the outputting of any two consecutive maximal cliques.
For this, the procedure firstly performs a bootstrapping phase, 
whose worst-case time complexity is bounded as follows:
\[\tau_{\texttt{boot}} = O(n^{\omega+3})=O\big(n^{5.3728639}\big).\]
Moreover, the procedure employs  
$O(n^{\omega - 2\omega(1,1,1/2) + 6}) = O(n^{4.2796})$ 
space.
\end{Thm}


In passing, we shall introduce a backtracking technique, named Batch Depth-First Search (Batch-DFS), 
whose aim is to keep the search of maximal cliques going on, 
solving one batch of problems after another, 
consuming only polynomial space overall. An in-depth time and space analysis of Batch-DFS is offered, as 
we believe that it may be of general interest for applying a similar approach to some other 
listing problems that admit polynomial time delay algorithms.

\paragraph{Organization.}
The rest of this paper is organized as follows.
In Section\ref{sect:background}, some background notation is 
introduced in order to support the subsequent sections.
Section~\ref{sect:RS-Tree} recalls some major aspects of 
Tsukiyama~\etal, Johnson~\etal, and Makino Uno's solutions; 
in particular, the construction of the RS-tree $\T_G$ is recalled and revised, 
as this is actually the enumeration tree of all the maximal cliques that we aim to list.
In Section~\ref{sect:reduction}, we describe our reduction from the problem of 
computing all children of any batch of $n^2$ nodes of $\T_G$ to that of performing rectangular matrix products.
The Batch-DFS backtracking is introduced and analyzed in Section~\ref{sect:backtracking}.
Our Maximal Clique Listing algorithms are offered in Section~\ref{sect:algorithm}.
Finally, Section~\ref{sect:conclusion} closes the paper.

%% file: clique-listing-Sect2-BandN.tex
\section{Background and Notation}\label{sect:background}

\begin{figure}[]
\center
\begin{tikzpicture}
  \begin{scope}
    \foreach \a in { 18, 90, 162, 234, 306 } {
      \foreach \b in { 18, 90, 162, 234, 306 } {
        \drawPolarLinewithBG{\a:2}{\b:2};
      }
    }
    \drawLinewithBG{0,2}{0,4};
    \drawLinewithBG{1.9,0.7}{1.9,3};
    \drawLinewithBG{-1.9,0.7}{-1.9,3};
 
    \drawLinewithBG{-1.9,3}{1.9,3};
    \drawLinewithBG{-1.9,3}{0,4};
    \drawLinewithBG{0,4}{1.9,3};
    
    \node[label={$6$}] at (0, 4) [circle,fill=black] {};
    \node[label={$7$}] at (1.9, 3) [circle,fill=black] {};    
    \node[label={$8$}] at (-1.9, 3) [circle,fill=black] {};    

    \node[label={right:$2$}] at (19:2cm) [circle,fill=black] {};
    \node[label={above left:$1$}] at (90:2cm) [circle,fill=black] {};
    \node[label={left:$5$}] at (162:2cm) [circle,fill=black] {};
    \node[label={left:$4$}] at (234:2cm) [circle,fill=black] {};
    \node[label={right:$3$}] at (306:2cm) [circle,fill=black] {};
  \end{scope}
\end{tikzpicture}
\caption{An example graph obtained by gluing together the complete graphs $\K_5$ and $\K_3$. 
The corresponding maximal cliques are $K_0=\{1, 2, 3, 4, 5\}$, $K_1=\{1, 6\}$, 
$K_2=\{2, 7\}$, $K_3=\{5, 8\}$, $K_4=\{6, 7, 8\}$. 
Here, $1<2<3<4<5<6<7<8$ and $K_0>_{\text{lex}} K_1 >_{\text{lex}} 
K_2 >_{\text{lex}} K_3 >_{\text{lex}} K_4$.}\label{fig:k5}
\end{figure}
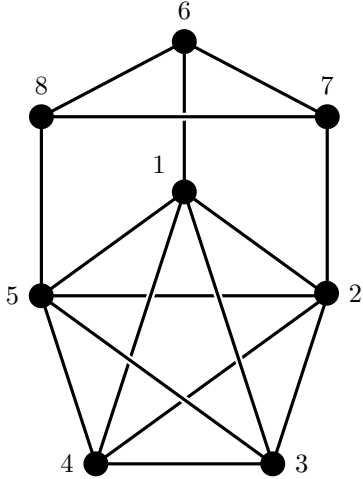

To begin with, our graphs are \emph{undirected} and \emph{simple}, \ie they have no self-loops nor parallel-edges. 
Let $[n]=\{1, \ldots, n\}$ for every $n\in\N$. Let $G = (V,E)$ be a graph with  
vertex set $V = [n]$ and edge set $E = \{e_1,\ldots,e_m\}$. Here, $|E|=m$ and $|V|=n$.
Moreover, for any vertex subset $S \subseteq V$, let $x(S)$ be the \emph{characteristic vector} of $S$, i.e., for every $i\in [n]$ 
the $i$-th coordinate of $x(S)$ is $1$ if $i \in S$, and it is $0$ otherwise.
For any vertex $v\in V$ of $G$, let $\Gamma(v)=\{u\in V \mid \{u,v\}\in E\}$ be the \emph{neighbourhood} of $v$, 
and let $\delta(v)=|\Gamma(v)|$ be the degree of $v$. Let $\Delta=\max_{v\in V}\delta(v)$ be the maximum degree of $G$. 
For any vertex subset $S\subseteq V$ and any index $i\in [n]$, define $S_{\leq i} = S\cap [i]$ 
and $S_{< i} = S\cap [i-1]$ (where $S_{<1}=\emptyset$). For any two vertex sets $X$ and $Y$, 
we say that $X$ is \emph{lexicographically greater} than $Y$, denoted $X>_{\text{lex}}Y$, 
if the smallest vertex (\ie the smallest natural number $i$) in the symmetric difference 
$(X\setminus Y) \cup (Y \setminus X)$ is contained in $X$. The usual common ordering on $\N$ is denoted $<$ 
(\ie without the subscript \texttt{lex}). A \emph{clique} is any subset $K$ 
of the vertex set $V$ such that any two vertices in $K$ are adjacent. 
A clique is \emph{maximal} when it is not a subset of any larger clique. For any clique $K$ (not necessarily a maximal one), 
let $\closure(K)$ be the \emph{lexicographic completion} of the clique $K$, namely, 
the lexicographically greatest among all the maximal cliques containing $K$. 
It is clear from its definition that $\closure(K)$ is not lexicographically smaller than $K$. 
To conclude this section, let $K_0$ be the maximal clique which is the lexicographically greatest among all the maximal cliques of $G$. 
Notice $K_0=\closure(\emptyset)$. This notation is exemplified in \figref{fig:k5}, where a running example for this paper is proposed.

%% file: clique-listing-Sect3-RS-Tree.tex
\section{The RS-Tree of Maximal Cliques}\label{sect:RS-Tree}
This section recalls some major aspects of the previous algorithms for {\mainproblem}, 
which were devised by Tsukiyama~\etal, Johnson~\etal, Makino and Uno, as these comprise the backstage and the backbone of our present solution.
In particular, this section reworks the construction of the \emph{Reverse Search Tree} (\emph{RS-tree}) $\T_G$ for enumerating 
all the maximal cliques of any given graph $G$. This is done by studying the corresponding parent-child relations.
In the original paper of Makino and Uno, all proofs about the characterization of $\T_G$ were omitted due to space restrictions. 
In the present work full proofs are presented for the sake of completeness.
Indeed, offering a simple and self-contained exposition of what in~\cite{MU04} was one of our purposes.
In cleaning out the arguments, and to help the understanding of the reader,
we opted for restructuring also the statements and the network of their relations. 
To begin with, let us observe some introductory properties. 
\begin{Prop}\label{prop:closureproperties} 
Let $K$ and $K'$ be two cliques of $G=(V,E)$. If $K\subseteq K'$ then $\closure(K)\geq_{\text{lex}}\closure (K')$.
\end{Prop}
\begin{proof}
Notice that $K\subseteq K'\subseteq \closure(K')$ and recall
that $\closure(K)$ is the lexicographically greatest 
maximal clique containing $K$.
\end{proof}

We proceed by observing a simple characterization of $\closure(\cdot)$.
\begin{Prop}\label{prop:char_closure} 
Let $K$ be a clique of $G=(V,E)$. For any $v\in [n]$, precisely one of the following two must occur:
\begin{enumerate}
\item $v\in\closure(K)$; 
\item there exists $z\in K\cup \big([v-1]\cap\closure(K)\big)$ such that $v\not\in\Gamma(z)$. 
\end{enumerate}
\end{Prop} 
\begin{proof} 
It is sufficient to show that 
$v\not\in\closure(K)$ if and only if (2) holds on $v$.
Since $\closure(K)$ is the lexicographically greatest maximal clique containing $K$, 
then $v\not\in\closure(K)$ if and only if at least one of the following two conditions hold:
either $v$ is not adjact to all vertices in $K$ (\ie there exists $z\in K$ such that 
$v\not\in\Gamma(z)$), or $v$ is not adjacent to some $z\in\closure(K)$ which is smaller than $v$ 
(\ie there exists $z\in [v-1]\cap\closure(K)$ such that $v\not\in\Gamma(z)$).
For this reason, $v\not\in\closure(K)$ if and only if there exists   
$z\in K\cup \big([v-1]\cap\closure(K)\big)$ such that $v\not\in\Gamma(z)$.
\end{proof}

The next proposition shows that $\closure(\cdot)$ is computable within $O(m)$ time.
\begin{Prop}\label{compute-C}
Let $K$ be a clique of any given graph $G=(V,E)$, where $|V|=n$ and $|E|=m$. 
The lexicographical completion $\closure(K)$ is computable within $O(\text{min}\{m, \Delta^2\})$ time. 
\end{Prop}
\begin{proof} 
Consider Algorithm~\ref{ALGO:compute-C}. It takes as input a clique $K$ of $G$.
Moreover, it employs the subprocedure \texttt{is-complete()} in order to test, 
on input $(u,X)$ for some $u\in V$ and $X\subseteq V$, whether $\{u,x\}\in E$ for every $x\in X$. 
This check can be done in $O(\delta(u))$ time. So, Algorithm~\ref{ALGO:compute-C} works as follows.
Firstly, an auxiliary set $S$ gets initialized as $S=K$ at line~1. Soon after, a vertex $\hat{v}\in K$ is picked up arbitrarily at line~2.
In the rest of the algorithm the auxiliary set $S$ will be incremented. 
The rationale, here, is that every node taking part to this augmentation must be among the neighbours of $\hat{v}$.
In fact, at line~5, Algorithm~\ref{ALGO:compute-C} augments $S$ with vertex $u$ if and only if $u$ is 
the lexicographically greatest vertex (\ie the smallest natural number) which lies in  
$\closure(K)\setminus S$. At the end, $S$ is returned at line~6. Let $\hat{S}$ be the set outputted by Algorithm~1.
Notice that, for every $v\in [n]$, precisely one of the following conditions hold: 
either $v\in\hat{S}$ or there exists $z\in K\cup([v-1]\cap\hat{S})$ 
such that $v\not\in\Gamma(z)$. Thus, by Proposition~\ref{prop:char_closure}, $\hat{S}=\closure(K)$.
Of course, Algorithm~\ref{ALGO:compute-C} halts within time 
$O\big(\sum_{u_i\in\Gamma(v)} \delta(u_i)\big) = O\big(\text{min}\{m, \Delta^2\}\big)$. 

\smallskip
\begin{algorithm}[H]\label{ALGO:compute-C}
\caption{Computing the Lexicographical Completion $\closure(\cdot)$.}
\DontPrintSemicolon
\nonl \SetKwProg{Fn}{Procedure}{}{}
\Fn{$\texttt{compute-C}(K, G)$}{
\KwIn{A clique $K$ of $G=(V,E)$.}
\KwOut{The lexicographical maximal clique $\closure(K)$ containing $K$.}
$S\leftarrow K$;\tcp{initialize the set $S$ to $K$} 
$\hat{v}\leftarrow$ pick any vertex $v\in K$; \;
\ForEach{$u\in \Gamma(\hat{v})$ (in ascending order w.r.t $\N$)}{ 
 \If{ $u\not\in K$ \textbf{ and } \texttt{is-complete}($u$, $S$) = \textbf{true} }{
 $S\leftarrow S\cup\{ u \}$; \;
}
}
\Return{$S$}; \;
}
\end{algorithm} 
\smallskip

Here above, the pseudocode of Algorithm~\ref{ALGO:compute-C} closes the proof.
\end{proof}

Given any $n$-vertex graph $G=(V,E)$, for any maximal clique $C$ ($\neq K_0$) 
there exists at least one index (\ie one vertex) $i\in [n]$ such that $\closure(C_{< i})\neq C$.
Indeed, $\closure(C_{\leq 0})=K_0\neq C$. 
In virtue of this fact it makes sense to define the 
\emph{parent} of $C$ as $\parent(C)=\closure(C_{< i})$,  
provided that $i\in [n]$ is the greatest index satisfying $\closure(C_{< i})\neq C$. 
Such an index $i$ is called the \emph{index of $C$}, and it is also denoted by $i(C)$. 
As mentioned, these indices are well defined. Moreover, notice $\parent(C)>_{\text{lex}}C$, 
\ie the parent $\parent(C)$ of any maximal clique $C$ ($\neq K_0$) is 
not lexicographically smaller than $C$. 
This implies that the corresponding parent-child binary relation is acyclic and creates an in-tree, 
denoted $\T_G$, which is directed towards its root $K_0$. 
We say that $\T_G$ is the RS-tree of $G$.
Of course, the nodes of $\T_G$ corresponds to the maximal cliques of $G$ that we aim to list. 

\figref{fig:tree} depicts the RS-Tree $\T_G$ associated to the running example graph of \figref{fig:k5}.
Every node of $\T_G$ depicts a maximal clique of $G$ and its corresponding index.

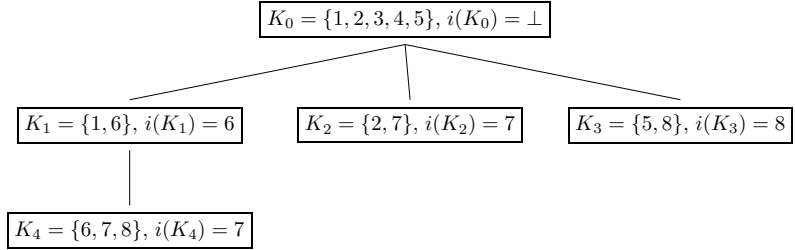
\begin{figure}[!htb]
    \centering
    \begin{tikzpicture}[scale=.8, level distance=50pt,sibling distance=15pt]
    \Tree [. \framebox{$K_0=\{1, 2, 3, 4, 5\}$, $i(K_0)=\bot$}
    	\edge node[left, xshift=0ex]{};
   	 [. \framebox{$K_1=\{1,6\}$, $i(K_1)=6$}
		\edge node[left, xshift=-1ex]{};
   	 	[. \framebox{$K_4=\{6,7,8\}$, $i(K_4)=7$}
  	 	 ]
  	  ]
    	\edge node[xshift=0ex]{};
   	 [. \framebox{$K_2=\{2,7\}$, $i(K_2)=7$}
  	  ]
 	   \edge node[right, xshift=0ex]{};
	    [. \framebox{$K_3=\{5,8\}$, $i(K_3)=8$}
	    ]
    	]
    \end{tikzpicture}
    \caption{The RS-Tree $\T_G$ corresponding to the running example graph $G$ of \figref{fig:k5}.}
\label{fig:tree}
\end{figure}

At this point, we shall provide an effective algorithm for computing $\parent(\cdot)$. 
Indeed, it is not difficult to see that $\parent(C)$ is computable from $C$ in linear $O(m+n)$ time. 
Here below, Proposition~\ref{compute-i} shows how to compute the index $i(\cdot)$, 
while Proposition~\ref{compute-P} finally provides an $O(m+n)$ time algorithm for computing the parent relation $\parent(\cdot)$. 
\begin{Prop}\label{compute-i} 
Let $C$ be a maximal clique of $G=(V,E)$, where $|V|=n$ and $|E|=m$. 

The index $i(C)$ 
is computable within $O(m+n)$ time.
\end{Prop}
\begin{proof}
Consider Algorithm~\ref{ALGO:compute-i}, it takes as input a maximal clique $C$ of $G$ and it aims to 
return the corresponding index $i(C)$ as output. The procedure works as follows: at the beginning,
each vertex $v\in V\setminus C$ is marked as \texttt{active}. Moreover, the procedure keeps track of a counter $d_{C}(v):V\rightarrow \N$, 
which is initialized to be the degree of $v$ with respect of $C$, for each $v\in V$.
Then, for each $v\in V$ in descending ordering from $n$ to $1$, 
Algorithm~\ref{ALGO:compute-i} checks whether ``$i(C)=v$" in the following manner:
\begin{inparaenum} 
\item if $v\not\in C$, then $v$ becomes \texttt{deactive} at line~14;
\item otherwise $v\in C$, then $v$ is (roughly speaking) turned-off within $C$, 
and thus the counter of every $u\in \Gamma(v)\setminus C$ is decremented at line~7.
The size of $C$ is also decremented at line~8. At this point, if there exists $z\in V\setminus C$ which is still \texttt{active} 
and such that $d_C(z)\geq\texttt{size}$, then $v$ is returned as output at line~12.
The existence of such $z$ can be checked quite efficiently as follows: at line~9 the procedure 
picks the greatest vertex (\ie the greatest natural number) $u\in C$ such that $u<v$, say $\hat{u}$, then, 
at line~10 and line~11, the neighbourhood of $\hat{u}$ is inspected in order to 
check whether there is any $z\in \Gamma(\hat{u})\setminus C$ which is still \texttt{active} and such that $d_C(z)\geq\texttt{size}$.
\end{inparaenum}
Also notice that, if line~12 is never reached, the procedure returns $\bot$ at line~15 
(because the index of the root $K_0$ is undefined). 

This concludes the description of Algorithm~\ref{ALGO:compute-i}.

\smallskip
\begin{algorithm}[H]\label{ALGO:compute-i}
\caption{Computing the Index $i(\cdot)$.}
\DontPrintSemicolon
\nonl \SetKwProg{Fn}{Procedure}{}{}
\Fn{$\texttt{compute-i}(C)$}{
\KwIn{a maximal clique $C$ of $G=(V,E)$.}
\KwOut{the index $i(C)$ of $C$.}
$d_C(v)\leftarrow|\{u\in C\mid \{u,v\}\in E\}|$; \tcp{\ie the degree of $v\in V$ in $C$.}
$\texttt{size} \leftarrow |C|$;\;
label each $v\in V\setminus C$ as \texttt{active};\;
\ForEach{ $v\in V$, from $n$ to $1$ }{
\If{ $v \in C$ }{
\ForEach{ $u\in\Gamma(v)\setminus C$}{
$d_C(u)\leftarrow d_C(u)-1$;\;
}
$\texttt{size} \leftarrow \texttt{size}-1$;\;
$\hat{u}\leftarrow $ the greatest vertex $u\in C$ such that $u < v$;\;
\ForEach{$z\in \Gamma(\hat{u})\setminus C$}{
\If{$z$ is \texttt{active} and $d_C(z)\geq\texttt{size}$}{
\Return{$v$;}
}
}
}
\Else{
label $v$ as \texttt{deactive};\;
}
}
\Return{$\bot$};
}
\end{algorithm} 
\smallskip

The correctness of the procedure follows easily from the definition of 
lexicographic completion $\closure(\cdot)$ and that of index $i(\cdot)$. 
Concerning its time complexity, observe that the procedure visits each vertex and each edge at most $O(1)$ times, 
and the work done at each one of those is $O(1)$ as well, so that Algorithm~\ref{ALGO:compute-i} always halts within $O(m+n)$ time.
\end{proof} 

\begin{Prop}\label{compute-P} 
Let $C$ be a maximal clique of $G=(V,E)$, where $|V|=n$ and $|E|=m$. 

The parent maximal clique $\parent(C)$ of $C$ is computable within $O(m+n)$ time. 
\end{Prop} 
\begin{proof}
Firstly, compute the index $i(C)$ with Algorithm~\ref{ALGO:compute-i} 
(defined in Proposition~\ref{compute-i}). 
Secondly, compute $\parent(C)=\closure(C_{i(C)})$ by invoking Algorithm~\ref{ALGO:compute-C}
(defined in Propostion~\ref{compute-C}). 
\end{proof}

As already mentioned, our algorithm, that of Tsukiyama~\etal~\cite{TsukiyamaIAS77}, 
Johnson~\etal~\cite{Johnson88}, as well as that of Makino and Uno~\cite{MU04}, 
traverse the nodes of $\T_G$ in a DFS-like fashion starting from the root $K_0$.
However, in order to traverse $\T_G$, we first need to show how to effectively characterize  
all the children $C$ of any given maximal clique $P$ of $G$. 
The following is a simple but crucial observation. 
In order for $P$ to be the parent of $C$, 
two \emph{reconstructability} conditions should hold at the same time, 
namely: \begin{enumerate} 
\item the parent $P$ should be ``reconstructible" from its child $C$; 
\item the child $C$ should be ``reconstructible" from its parent $P$ and index $i(C)$.
\end{enumerate}

On this way, the following fact turns out to play a twofold pivotal role.
\begin{Lem}[Reconstructability Lemma]\label{lem:rec} Let $P$ and $C$ be maximal cliques of $G$. 
Assume that $P=\parent(C)$ and let $i$ be the index of $C$. Then, $C_{< i} = P_{< i}\cap \Gamma(i)$. 
\end{Lem}
\begin{proof} 
\begin{itemize}
\item Firstly, we argue $C_{< i} \subseteq P_{< i}\cap\Gamma(i)$.

Since $i$ is the index of $C$, then $i\in C$. Thus, $C_{< i}\subseteq \Gamma(i)$. 
Moreover, $P=\P(C)=\closure(C_{< i})$ implies $C_{< i}\subseteq P_{< i}$.
\item Secondly, we argue $P_{< i}\cap \Gamma(i)\subseteq C_{< i}$.

Let $v$ be any node in $P_{<i}\cap \Gamma(i)$.
Notice $v<i$ and recall that $C=\closure(C_{\leq i})$.
In order to show $v\in C_{<i}$, it is thus sufficient to prove $v\in \closure(C_{\leq i})$.
For this, we shall rely on Proposition~\ref{prop:char_closure}. 
Now, observe the following two facts:
\begin{enumerate}
\item No $z\in C_{\leq i}$ is such that $v\not\in\Gamma(z)$. 
In fact, $v$ is adjacent to all vertices in $C_{<i}$ (because $v\in P=\closure(C_{<i})$) and 
$v$ is adjacent to $i$ (because $v\in \Gamma(i)$).
\item No $z\in [v-1]\cap \closure(C_{\leq i})$ is such that $v\not\in\Gamma(z)$.
In fact, since $P$ is a clique containing $v$, no $z\in [v-1]\setminus \Gamma(v)$ belongs to $P$, 
whence neither to $C_{<i}$ (because $C_{< i}\subseteq P_{< i}\subseteq P$), 
nor to $C_{\leq i}$ (because $v<i$). This implies that no 
$z\in [v-1]\setminus \Gamma(v)$ belongs to $\closure(C_{\leq i})$, 
because $\closure(C_{\leq i})=C$ and $v<i$.
\end{enumerate}
At this point, $v\in\closure(C_{\leq i})$ follows directly from Proposition~\ref{prop:char_closure}. 

\end{itemize}
\end{proof}

We are now in position to characterize the parent and children reconstructability conditions. 
Indeed, at this point, they both turn out to be a direct consequence of Lemma~\ref{lem:rec}.
\begin{Prop}[Parent and Child Reconstructability]\label{prop:rec}
Let $P$ and $C$ be maximal cliques of $G$. Assume that $P=\parent(C)$ and let $i$ be the index of $C$.
Then, the following two conditions hold: 
\begin{enumerate}
\item $P=\closure\big( C_{< i} \big) 
	= \closure\big( P_{< i}\cap \Gamma(i) \big)$ \;\;\;\; (Parent Reconstructability)
\item $C=\closure\big(( P_{< i}\cap \Gamma(i) )\cup \{i\} \big)$ \;\;\;\; (Child Reconstructability)
\end{enumerate}
\end{Prop}
\begin{enumerate}
\item[\emph{Proof of 1.}] It is sufficient to observe the following:
\begin{align*}
P & = \closure\big( C_{< i} \big) & \text{ (because $i=i(C)$ and $P=\parent(C)$)} \\
  & = \closure\big( P_{< i}\cap \Gamma(i) \big) & 
				\text{(by Lemma~\ref{lem:rec})} 
\end{align*}
\item[\emph{Proof of 2.}] It is sufficient to observe the following:  
\begin{align*}
C & = \closure\big(C_{< i}\cup \{i\}\big) & \text{ (because $i=i(C)$)}  \\
& = \closure\big(( P_{< i}\cap \Gamma(i))\cup \{i\} \big) 	
	& \text{ (by Lemma~\ref{lem:rec})} 
\end{align*}

\end{enumerate}

The rationale which allows for the computation of a maximal clique child $C$ is that of \emph{reversing} the parent relation $\parent(\cdot)$, 
in the spirit of the Reverse Search Enumeration of Avis and Fukuda~\cite{AvisFukuda93}. 
Observe that Item~2 of Proposition~\ref{prop:rec} pointed out a shape for such an inversion. 
In fact, in light of Proposition~\ref{prop:rec}, 
given any maximal clique $P$ of $G$ and any $i\in [n]$,  
it is natural to introduce the following notation:
\[ 
\C(K,i) = \closure\big(( K_{< i}\cap \Gamma(i) )\cup \{i\} \big). 
\] 
Proposition~\ref{prop:rec} tells us that whenever $C$ is a child of $P$ with index $i$, then $C=\C(P,i)$. 
This means that, given $P$, we are called to characterize all the indices $i\in [n]$ such that $\C(P, i)$ is a child of $P$ with index $i$. 
In order to do that, let us proceed by observing the following property enjoyed by $\closure(\cdot)$, 
it will turn out to be pertinent in a while. 
\begin{Lem}\label{lemma:Cidemp} 
Let $G$ be any $n$-vertex graph. 
Let $K$ be a clique of $G$ and let $a,b\in [n]$ be two indices such that $a\leq b$. 
Then, $\closure(\closure(K_{\leq a})_{\leq b})=\closure(K_{\leq a})$; 
\end{Lem}
\begin{proof}
Since $a\leq b$, then $K_{\leq a}\subseteq \closure(K_{\leq a})_{\leq b}$. 
Thus, by Proposition~\ref{prop:closureproperties}, 
\[\closure(K_{\leq a})\geq_{\text{lex}} \closure(\closure(K_{\leq a})_{\leq b}).\]
On the other way, $\closure(K_{\leq a})_{\leq b}\subseteq \closure(K_{\leq a})$. 
Thus, by Proposition~\ref{prop:closureproperties} again, 
\[\closure(\closure(K_{\leq a})_{\leq b})\geq_{\text{lex}} 
\closure(\closure(K_{\leq a}))=\closure(K_{\leq a}).\]

Since $\geq_{\text{lex}}$ is a total ordering, the observations 
above imply $\closure(\closure(K_{\leq a})_{\leq b})=\closure(K_{\leq a})$.
\end{proof} 

We are now in position to characterize all the children of any given maximal clique $P$.
\begin{Prop}\label{prop:children_char} 
Let $P$ be a maximal clique of any given $n$-vertex graph $G=(V,E)$. 

There exist a child of $P$ having index $i$ if and only if 
$i\not\in P\cup [i(P)]$ and the following two reconstructability conditions hold: 
\begin{enumerate}
\item[a.] $P_{< i} = \closure\big(P_{< i}\cap\Gamma(i)\big)_{< i}$
\item[b.] $\closure\big(( P_{< i}\cap \Gamma(i) )\cup \{i\} \big)_{< i} = P_{< i} \cap \Gamma(i)$
\end{enumerate} 
\end{Prop}
\begin{proof} ($\Rightarrow$)
Let $C$ be the child of $P$ having index $i$ (which exists by assumption).

Firstly, we argue that $i\in P\cup [i(P)]$. 
Indeed, since $P=\parent(C)=\closure(C_{< i})$ and $C=\closure(C_{\leq i})$, then $i\not\in P$.
Moreover, the following equalities show that $i>i(P)$: 
\begin{align*}
\closure(P_{< i})&= \closure(\closure(C_{< i})_{< i})& \text{(by $P=\closure(C_{< i}$))} \\
               &= \closure(C_{< i}) & \text{(by Lemma~\ref{lemma:Cidemp})}\\
               &= P & \text{(by $P=\closure(C_{< i}$)}
\end{align*}
Finally, we argue that both the (\emph{a}) and the (\emph{b}) conditions hold on $i$.

\begin{enumerate}
\item[\emph{Proof of a.}] By Item~1 of Propostion~\ref{prop:rec}, 
we have $P=\closure(P_{<i}\cap\Gamma(i))$. Thus, (\emph{a}) holds on $i$.
\item[\emph{Proof of b.}] By Item~2 of Proposition~\ref{prop:rec}, 
we have that $C=\closure\big(( P_{< i}\cap \Gamma(i) )\cup \{i\} \big)$.
By Lemma~\ref{lem:rec}, we have $C_{< i} = P_{< i} \cap \Gamma(i)$. 
These facts imply that the (\emph{b}) condition holds on $i$.
\end{enumerate}

($\Leftarrow$) We argue that whenever both the (\emph{a}) 
and the (\emph{b}) conditions hold on some $i\not\in P \cup [i(P)]$, 
then there exist a child of $P$ with index $i$. 
Let $C=\closure\big(( P_{< i}\cap \Gamma(i) )\cup \{i\} \big)$ for some $i$ as mentioned.
Firstly, observe that $C\neq P$: in fact, $i\in C$ by definition of $C$ 
but $i\not\in P$ by hypothesis.

Now, we argue that $P=\closure(C_{< i})$. 
In fact, observe that the following equalities hold: 
\begin{align*}
P &=  \closure(P_{< i}) & \text{(by $i \not\in [i(P)]$)} \\
  &=  \closure(\closure(P_{< i}\cap\Gamma(i))_{< i}) & \text{(by (\emph{a}))} \\
  &=  \closure(\closure(\closure\big(( P_{< i}\cap \Gamma(i) )\cup 
\{i\} \big)_{< i})_{< i}) & \text{(by (\emph{b}))} \\
  & = \closure(\closure(C_{< i})_{< i}) & \text{(by definition of $C$)} \\
  &=  \closure(C_{< i}) & \text{(by Lemma~\ref{lemma:Cidemp})}
\end{align*}

To conclude the proof, it is sufficient to check that $\closure(C_{\leq i})=C$.
\begin{align*}
\closure(C_{\leq i}) & = \closure(C_{< i}\cup \{i\}) & \text{(because $i\in C$)} \\
           & = \closure( \closure\big(( P_{< i}\cap \Gamma(i) )\cup 
\{i\} \big)_{< i}\cup \{i\}) & \text{(by definition of $C$)} \\
	   & = \closure\big((P_{<i}\cap\Gamma(i)) \cup \{i\}\big) & \text{(by (\emph{b}))} \\
           & = C & \text{(by definition of $C$)} \\
\end{align*}
\end{proof}

Observe that, since $\closure(K)$ can be computed from any clique $K$ in $O(m)$ time by Proposition~\ref{compute-C}, 
it is possible to compute all the children of a given maximal clique $P$ in $O(mn)$ time by Proposition~\ref{prop:children_char}.
In fact, it is sufficient to check whether the conditions (\emph{a}) and (\emph{b}) both hold on the index $i$, 
for each $i\in [n]\setminus (P\cup [i(P)])$. In this manner, listing each node of $\T_G$ (namely, 
each maximal clique of $G$) with $O(mn)$ time delay. In order to improve over the $O(mn)$ bound, 
Makino and Uno reduced the problem of checking the conditions (\emph{a}) and 
(\emph{b}) of Proposition~\ref{prop:children_char} to that of multiplying two $n\times n$ square matrices~\cite{MU04}. 
In doing this, they observed (without proof) the following two lemmata. These are a restating of the conditions 
(\emph{a}) and (\emph{b}) of Proposition~\ref{prop:children_char}. 
We remark that Lemma~\ref{lemmaA} and Lemma~\ref{lemmaB} are really at the ground of our reduction to rectangular matrix multiplication.

\begin{Lem}\label{lemmaA} 
Let $P$ be a maximal clique of any given $n$-vertex graph $G=(V,E)$. 

Then, $i\in [n]$ satisfies $P_{< i}=\closure\big(P_{< i}\cap\Gamma(i)\big)_{< i}$ 
if and only if there doesn't exist any index $j\in [i-1]\setminus P$ such that the 
following conditions hold:
\begin{enumerate}
\item[a'.] $j$ is adjacent to all vertices in $P_{<j }$;
\item[a''.] $j$ is adjacent to all vertices in $P_{< i}\cap\Gamma(i)$.
\end{enumerate}
\end{Lem}
\begin{proof}
Assume that for some $i\in[n]$ there exists  $j\in [i-1]\setminus P$ 
satisfying both the (\emph{a'}) and (\emph{a''}) condition.
Then, there exists $j'\leq j$ such that 
$j'\in \closure(P_{< i}\cap \Gamma(i))_{< i}\setminus P_{< i}$, 
thus implying $\closure(P_{< i}\cap \Gamma(i))_{< i}\neq P_{< i}$.
For the opposite direction, assume that for some $i\in [n]$ there is no $j\in [i-1]\setminus P$ 
satisfying both the (\emph{a'}) and (\emph{a''}) condition. 
Then, $\closure\big(P_{< i}\cap\Gamma(i)\big)_{< i}=P_{< i}$ 
follows by definition of $\closure(\cdot)$. This concludes the proof.
\end{proof}

\begin{Lem}\label{lemmaB}
Let $P$ be a maximal clique of any given $n$-vertex graph $G=(V,E)$. 

Then, $i\in [n]$ satisfies $\closure\big(( P_{< i}\cap \Gamma(i) )\cup 
\{i\} \big)_{< i}=P_{< i}\cap\Gamma(i)$ if and only if 
there doesn't exist any index $j\in [i-1]\setminus (P_{< i}\cap\Gamma(i))$ 
such that the following condition hold:
\begin{enumerate}
\item[b'.] $j$ is adjacent to all vertices in $\big(P_{< i}\cap\Gamma(i)\big)\cup \{i\}$.
\end{enumerate}
\end{Lem}
\begin{proof}
Let $C=\closure\big(( P_{< i}\cap \Gamma(i) )\cup \{i\} \big)$ for some $i\in [n]$. 
Assume that there exists $j\in [i-1]\setminus (P_{< i}\cap\Gamma(i))$ satisfying 
the (\emph{b'}) condition. Then, there exists $j'\leq j$ such that 
$j'\in C_{< i}\setminus \big(P_{< i}\cap\Gamma(i)\big)$. This implies 
$C_{< i}\neq \big(P_{< i}\cap\Gamma(i)\big)$. For the opposite direction, 
assume that there is no $j\in [i-1]\setminus (P_{< i}\cap\Gamma(i))$ satisfying 
the (\emph{b'}) condition. 
Then, $C_{< i}=P_{< i}\cap\Gamma(i)$
follows by definition of $C$ and $\closure(\cdot)$.
This implies that $C_{< i} = P_{\leq i}\cap\Gamma(i)$ 
if and only if there is no $j\in [i-1]\setminus \big(P_{\leq i}\cap\Gamma(i)\big)$ 
satisfying the (\emph{b'}) condition. 
\end{proof}

%% file: clique-listing-Sect4-Reduction.tex
\section{Reduction to Rectangular Matrix Multiplication}\label{sect:reduction}

Given any maximal clique $P$ of $G=(V,E)$,  consider the problem of computing all the indices $i\in [n]\setminus (P\cup [i(P)])$ 
such that $\C(P,i)$ is a child of $P$ with index $i$. By Proposition~\ref{prop:children_char}, this amounts to check, for each 
$i\in [n]\setminus (P\cup [i(P)])$, whether both the conditions (\emph{a}) and (\emph{b}) hold on $i$ with respect to $P$. 
So, let us denote by $I_a^{P}$ and $I_b^{P}$ the sets of indices $i\in [n]\setminus (P\cup [i(P)])$ that 
satisfy the conditions (\emph{a}) and (\emph{b}) (respectively) of 
Proposition~\ref{prop:children_char} for some given maximal clique $P$ of $G$.  
Recall that the index $i(P)$ can be computed from $P$ within $O(n + m)$ time (by Proposition~\ref{compute-i}). 
The most expensive step is thus the computation of both $I_a^{P}$ and $I_b^{P}$ 
(which, recall, can always be done within $O(mn)$ time 
by performing at most $n$ computations of the lexicographical completion $\closure(\cdot)$).
Also, recall that this computation can be performed by checking the equivalent 
conditions (\emph{a'}), (\emph{a''}) for $I_a^P$ and (\emph{b'}) 
for $I_b^P$ given by  Lemma~\ref{lemmaA} and Lemma~\ref{lemmaB} (respectively).
As already mentioned in the previous section, in order to compute $I_a^{P}$ and $I_b^{P}$, 
Makino and Uno relied on fast square matrix multiplication, 
thus sharpening Tsukiyama's $O(mn)$ bound to $O(n^{\omega})$~\cite{MU04}.

At this point, we shall diverge from their approach in the following manner. 

\smallskip
\textbf{An Overview.} 
We denote by $\I^P$ the problem of computing the sets $I^P_a$ and $I^P_b$, 
with respect to some given maximal clique $P$ of any given $n$-vertex graph $G=(V,E)$. Moreover, we denote by $\B=\{P_1, \ldots, P_{|\B|}\}$ 
any batch (\ie family) of pairwise distinct maximal cliques of $G$. 
It is quite natural at this point to consider the problem $\I^\B$, namely, that of solving $\I^P$ for every $P\in\B$.
The intuition underlying our approach goes as follows: instead of solving each problem instance $\I^P$ separately (one after another,  
by reducing it to square matrix multiplication as in \cite{MU04}),  
we propose to group multiple maximal cliques into batches $\B$ and to solve the corresponding problem $\I^\B$ (for each batch $\B$), 
in one single shot, by reducing it to that of multiplying two rectangular matrices of size 
$|\B|\times n$ and $n\times n^2$. This rectangular matrix product 
can be performed in an asymptotically efficient way by adopting the algorithms devised by Le~Gall in~\cite{Gall12}.
As we will show in the forthcoming, the optimal size of $\B$ turns out to be $|\B|=|V|^2=n^2$. 
For this reason, we are going to deal with $n^2\times n$ by $n\times n^2$ matrix products.

\smallskip
\textbf{The Reduction.}

By virtue of Proposition~\ref{prop:children_char}, Lemma~\ref{lemmaA} and Lemma~\ref{lemmaB}, 
the problem of solving $\I^\B$ boils down in a straightforward manner   
to that of solving the following ‘‘\emph{kernel}’’ problem, which is denoted $\K^\B$ and defined in this way:
given $\B$ as input, for every maximal clique $P\in \B$, and for each pair of indices $(i,j)\in [n]\times [n]$, 
decide whether $(i,j)$ is \emph{good} with respect to $P$, namely, 
decide whether there exists $u\in P_{< i}\cap \Gamma(i)$ such that $\{u, j\}\not\in E$.
A \emph{solution} of $\K^\B$ is a mapping which assigns to each $P\in\B$ 
a boolean vector, denoted $[g^P_{ij}]_{ij}$ (where $g^P_{ij}\in\{\top,\bot\}$ for every $i,j\in [n]$), 
such that: 
\[
g^P_{ij}=
\left\{\begin{array}{ll}
\top, & \text{ if the pair } (i,j) \text{ is \emph{good} with respect to } P; \\ 
\bot, & \text{ otherwise.}
\end{array}\right.
\]

The rationale at the ground of these definitions 
clearly lies within Lemma~\ref{lemmaA} and Lemma~\ref{lemmaB}.
In fact, with these lemmata in mind, a moment's reflection reveals that once  
$\B\mapsto [g^P_{ij}]_{ij}$ has been determined, 
then, for each $P\in \B$, it is possible to solve $\I^P$ within time $O(n^2)$.

In summary, it is not difficult to see that these arguments allow one to solve $\I^\B$ within time: 
\[\texttt{Time}\big[\I^\B\big] = O\big(\texttt{Time}\big[\K^\B\big] + |\B|\,n^2\big).\]

Indeed, solving $\K^\B$ turns out to be the time bottleneck for solving $\I^B$.
The following proposition finally shows how to reduce $\K^\B$ to the problem 
of multiplying two rectangular matrices of size $|\B|\times n$ and $n\times n^2$.
\begin{Prop}[Reduction to Rectangular Matrix Multiplication]\label{prop:reduction}
Let $\B=\{P_1,\ldots, P_{|\B|}\}$ be a batch of maximal cliques of any given $n$-vertex graph $G=(V,E)$. 
Consider the $|\B|\times n$ matrix $M_\B$ whose $k$-th row, denoted $x_k$ for every $k\in [|\B|]$, 
is the characteristic vector $x_k=x(P_k)$.

For every $i,j\in [n]$, define the following subsets of $V$: 
\[A_i=V_{< i}\cap \Gamma(i) \text{ and } B_j=\Gamma(j).\] 
Let $M_G$ be the $n\times n^2$ matrix   
whose $(i,j)$-th column is the characteristic vector 
$x_{i,j}=x(A_i\setminus B_j)$.

Let $M_{\B,G}$ be the $|\B|\times n^2$ matrix obtained by performing the matrix product:
 \[M_{\B,G}=M_\B\, M_G.\]

For every $k\in[|\B|]$ and $i,j\in [n]$, 
denote by $M_{\B,G}[k, (i,j)]$ the particular entry of $M_{\B,G}$ whose row index is $k$ and 
whose column index is $(i,j)$. 
Finally, let us define: 
\[
g^P_{ij}=
\left\{\begin{array}{ll}
\top, & \text{ if } M_{\B, G}[k, (i,j)] > 0; \\ 
\bot, & \text{ otherwise}. \\
\end{array}\right.
\]
Then, the mapping $P_k\mapsto [g^{P_k}_{ij}]_{ij}$, 
which is defined for every $P_k\in\B$, is a correct solution of $\K^\B$.
\end{Prop}
\begin{proof}
To start with, let us fix $P\in\B$ and $i,j\in [n]$ arbitrarily.
Observe that $j\in V$ is adjacent to all the vertices in $P_{< i}\cap \Gamma(i)$ if and only if 
$\big( P_{< i}\cap \Gamma(i)\big) \setminus \Gamma(j)=\emptyset$. Equivalently, \[(i,j)\in [n]\times[n] \text{ is good \wrt $P$  } \iff
 P \cap\Big(\big( V_{< i}\cap \Gamma(i)\big) \setminus \Gamma(j)\Big)\neq\emptyset.\] 
Clearly, $\big( V_{< i}\cap \Gamma(i)\big) \setminus \Gamma(j)$ 
depends only on $i,j$ and not on $P$, so one can safely write this set as $A_i\setminus B_j$. 
Thus, in order to assess whether $(i,j)$ is good with respect to $P$, 
it is sufficient to check whether $P\cap (A_i\setminus B_j)\neq \emptyset$.
Let $k\in[|\B|]$ be the index of $P$ in $\B$, \ie assume that $P=P_k\in\B$. 
By definition of $M_\B, M_G$ and from the fact that $M_{\B, G}=M_\B\, M_G$, 
the following holds: \[P\cap(A_i\setminus B_j)\neq \emptyset \text{ if and only if } P_{\B, G}[k, (i,j)]>0.\]
This implies the thesis and concludes the proof. 
\end{proof}

\textbf{Time Complexity of the Reduction.}
We now focus on the time complexity of the reduction described in Proposition~\ref{prop:reduction}. 
To begin, as shown in Appendix~\ref{app:time_complexity}, 
and according to current upper bounds on rectangular fast matrix multiplication \cite{Gall12, Huang98}, 
the optimal size of the batch $\B$ turns out to be $|\B|=n^2$. 
Let us briefly retrace the argument that led us to this result, the full details are given in Appendix~\ref{app:time_complexity}.
Recall from Proposition~\ref{prop:reduction} that 
$M_{\B,G}$ can be computed by performing an $|\B|\times n$ by $n\times n^2$ matrix product.
Also recall that, by computing $M_{\B,G}$, one actually solves in one single shot $|\B|$ problem instances, 
\ie $\I^P$ for every $P\in\B$. Let $k\in\Q$ be such that $|\B|=\lfloor n^k\rfloor$. 
Then, computing $M_{\B,G}$, the \emph{amortized} time $\texttt{Time}\big[\I^{P}\big]$ 
for solving each problem $\I^P$ for $P\in B$ becomes the following,  
where $\texttt{Time}\big[M_{\B,G}\big]$ denotes the time it takes to compute $M_{\B}\cdot M_{G}$:
\begin{align*}
\texttt{Time}\big[\I^{P}\big] &   = O\Big(\frac{\texttt{Time}\big[\I^{\B}\big]}{|\B|}\Big) 
                                  = O\Big(\frac{\texttt{Time}\big[\K^{\B}\big]+|\B|n^2}{|\B|}\Big) \\ 
			      &   = O\Big(\frac{\texttt{Time}\big[M_{\B,G}\big]}{\lfloor n^k \rfloor} + n^2\Big) \\ 
			      &   = O\Big( n^{\omega(k, 1, 2)-k} + n^2\Big) \\  
\end{align*}
Our aim would be to find $k\in [0, +\infty)\cap\Q$ such that $\omega(k, 1, 2)-k$ attains its global minimum value. 
Even though the exact values of $\omega(k, 1, 2)$ are currently unknown, 
one can nevertheless minimize functions arising from state of the art upper bounds on $\omega(k, 1, 2)$.
These upper bounds have been derived within the framework of so-called \emph{bilinear} algorithms, see e.g.~\cite{Huang98, Gall12}.
For instance, in this work we consider the bound $f_{\text{HP98}}$ of Huang and Pan~\cite{Huang98}, 
and then the bound $f_{\text{LG12}}$ of Le~Gall~\cite{Gall12} 
(in particular, $f_{\text{LG12}}$ leads to the best upper bound on $\omega(2, 1, 2)$ which is currently known). 
Here, both estimates are be applied in such a way as to bound $\omega(k,1,2)-k$ from above.
The corresponding functions, that we aim to minimize, are denoted $g_1$ and $g_2$. 
Their behaviour is shown in \figref{fig:plot_bound}. 
\begin{figure}[h!] 
\centering
\begin{tikzpicture}

\newcommand\alphafm{0.30298}
\newcommand\omegafm{2.3728639}

\begin{axis}[xmax=5, ymax=3.5, xlabel={$k$}, 
legend entries={\text{$g_1(k)$, ref.\cite{Huang98}},
                \text{$g_2(k)$, ref.\cite{Gall12}}}, 
legend pos=outer north east
]
\addlegendimage{only marks,blue}
\addlegendimage{only marks,red}
  \addplot[
	blue, ultra thick,   
	domain= 0 : \alphafm,
	samples=50
  ] { 3-x };

  \addplot[
	blue, ultra thick,  
	domain= \alphafm : 1,
	samples=50
  ] { (2*(1-\alphafm) + (1-x) + (\omegafm -1)*(x-\alphafm))/(1-\alphafm)-x };

  \addplot[
	scatter,
	only marks, 
	point meta=explicit symbolic, 
	scatter/classes={
		a={mark=square*, blue} %
	},
   ] table[meta=label] {
	x   y          label 
  	1   2.334   a  
  };

  \addplot[
	blue, ultra thick, 
	domain= 1 : 2,
	samples=50
  ] { x*( ( (2/x)*(1-\alphafm) + (1-(1/x)) + (\omegafm-1)*((1/x)-\alphafm) ) / (1-\alphafm) ) - x };

  \addplot[
	blue, ultra thick,  
	domain= 2 : 5,
	samples=50
  ] { 2*( (x/2)*(1-\alphafm) + 0.5 + (\omegafm-1)*(0.5-\alphafm) ) / (1-\alphafm) - 0.99*x - 0.02 };

  \addplot[
	scatter,
	only marks, 
	point meta=explicit symbolic, 
	scatter/classes={
		a={mark=square*, red} %
	},
   ] table[meta=label] {
	x   y          label 
  	1   2.256689   a  
  };

  \addplot[red, dotted] coordinates {
  	(1, 2.256689) (1/0.5, 2.093362)   
  };

  \addplot[red, ultra thick, dashed] coordinates {
  	(1/0.5, 2.093362) (1/0.45, 2.2824489) (1/0.4, 2.5304375) (1/0.35, 2.8653429) (1/0.34, 2.9463853)
        (1/0.33, 3.0331485) (1/0.32, 3.1261594) (1/0.31, 3.2260097) (1/0.30298, 3.3005479)
  };

  \addplot[
	red, ultra thick, dashed,   
	domain= 1/0.30298 : 5,
	samples=50
  ] { x*2-x };

\end{axis} 
\end{tikzpicture} 
\caption{Plot of $g_1(k)$ for $k\in[0, +\infty)$, and of $g_2(k)$ for $k\in \{1\}\cup[2, +\infty)$.}\label{fig:plot_bound} 
\end{figure}
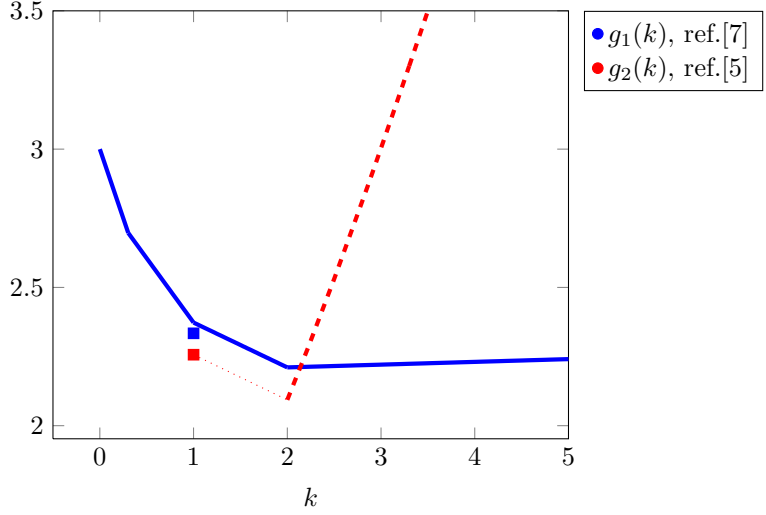 

The first function $g_1$ is defined as follows: \[g_1(k)= f_{\text{HP98}}(k)-k,\text{ for every } k\in[0, +\infty),\] 
here, $f_{\text{HP98}}:[0, +\infty)\rightarrow\RR$ 
is a piecewise-linear function, which was essentially pointed out by Huang and Pan in~\cite{Huang98}, 
and it satisfies $\omega(k,1,2)\leq f_{\text{HP98}}(k)$ for every $k\in [0, +\infty)\cap\Q$.
An analytic closed-form formula for $f_{\text{HP98}}$ is derived in Appendix~\ref{app:time_complexity}. 
Here, we just mention that $g_1(k)$ attains its global minimum value at $k = 2$, 
\ie \[ g_1(2)= \min_{k\in [0, +\infty)} g_1(k) = 2.2107878\] 
The qualitative behaviour of $g_1$ is traced in \figref{fig:plot_bound} with a filled blue colored line. 

In a similar way, the second function $g_2$ is defined as follows: 
\[ g_2(k) = k f_{\text{LG12}}(k)-k, \text{ for every } k\in \{1\}\cup [2, +\infty). \] 
here, $g_2$ takes into account the upper bound $f_{\text{LG12}}(k)$ for 
$\omega(1,1,1/k)$, which was established by Le~Gall~in~\cite{Gall12}.
We remark that, at the current state of art, the upper bounds of Le~Gall apply to $\omega(r,s,t)$ if and only if $r=s$.
For this reason, when $k\in [2,+\infty)\cap\Q$, we applied Le~Gall's bounds on 
$\omega(k,1,2)$ by relying on the following upper bound: 
\begin{align*}
\omega(k,1,2) & \leq  \omega(k,1,k) & \text{(for every $[2,+\infty)\cap\Q$)} \\ 
              &  =   k\,\omega(1,1,1/k) & \text{(by homogeneity)} \\
	      & \leq k f_{\text{LG12}}(k). &  
\end{align*}
In addition, we applied Le~Gall's bounds on $n\times n$ by $n\times n^2$ matrix products by considering the complexity 
exponent $\omega(1,1,2)$, which is actually one of those explicitly studied by Le~Gall~in~\cite{Gall12}. 
Notice that, when $k\in (1,2)$, it is not possible to apply the results of Le~Gall~\cite{Gall12} 
to bound $\omega(k,1,2)$, because $k\neq 1, k\neq 2$ and $1\neq 2$ so the above mentioned condition (\ie that $r=s$ in $\omega(r,s,t)$) doesn't apply in that case. 
This explains why $g_2(k)$ is defined on $k\in \{1\}\cup [2, +\infty)$. 
The qualitative behaviour of $g_2$ is traced in \figref{fig:plot_bound} with a dashed red colored line. 
Concerning the global minimization of $g_2$, it turns out\footnote{See Appendix~\ref{app:time_complexity} for the details.} that: 
\[
g_2(2) = \min_{k\in \{1\}\cup [2, +\infty)} g_2(k) = 2.093362.
\]
In summary, both estimates~\cite{Gall12}~and~\cite{Huang98} indicate that the minimum complexity comes at $k=2$, 
namely, they both indicate that the optimal size of the batch $\B$ is given by $|\B|=n^2$. 

From now on, let us fix the size of $\B$ to be $|\B|=n^2$. 
Then, by Proposition~\ref{prop:reduction},  
$M_{\B,G}$ can be computed by performing an $n^2\times n$ by $n\times n^2$ matrix product. 
Assuming $N=n^2$, it is equivalent to consider matrix products of 
type $N\times \lfloor N^{1/2} \rfloor$ by $\lfloor N^{1/2}\rfloor\times N $. As shown by Le~Gall in~\cite{Gall12}, 
the corresponding complexity exponent, which is $\omega(1,1,1/2)$, satisfies $\omega(1,1,1/2)\leq 2.046681$. 

Thus, $M_{\B,G}$ can be computed within the following time bound: 
\[\texttt{Time}\big[M_{\B,G}\big] = O\big(N^{\omega(1,1,1/2)}\big) = O\big(n^{2{\omega(1,1,1/2)}}\big) = O(n^{4.093362}).\] 
In this way we obtain: 
\begin{align*}
\texttt{Time}\big[\I^\B\big] & = O\Big(\texttt{Time}\big[\K^\B \big] + |\B|n^2 \Big) 
                               = O\Big(\texttt{Time}\big[M_{\B,G} \big] + |\B|n^2 \Big) \\ 
                             & = O\big(n^{2\omega(1,1,1/2)}\big)=O(n^{4.093362}).
\end{align*}

Then, each problem instance $\I^P$ for $P\in\B$ gets solved within the following time bound: 
\begin{align*}
\texttt{Time}\big[\I^{P}\big] & = O\Big(\frac{\texttt{Time}\big[\I^{\B}\big]}{|\B|}\Big) \\ 
                              & = O\big(n^{2\omega(1,1,1/2)-2}\big) = O(n^{2.093362}).
\end{align*}
which, we remark, it is an \emph{amortized} time bound across $n^2$ problem instances.

\paragraph{An Algorithm for Solving $\I^\B$.}
The pseudocode for solving an instance of $\I^\B$ by reducing it to rectangular matrix 
multiplication is shown below in Algorithm~\ref{ALGO:solve-IB}. 

\smallskip
\begin{algorithm}[H]\label{ALGO:solve-IB}
\caption{Solving $\I^\B$ by Reduction to Rectangular Matrix Multiplication.} 
\DontPrintSemicolon
\nonl \SetKwProg{Fn}{Procedure}{}{}
\Fn{\texttt{Solve\_Rectangular\_$\I$}$(\B)$}{
\KwIn{A batch of (exactly) $n^2$ maximal cliques $\B=\{P_1, \ldots, P_{n^2}\}$ of $G=(V,E)$.}
\KwOut{A vector $L_\B$ representing all children of every $P\in\B$, \ie $L_\B =\big\{(P, \texttt{list}_P) \mid P\in\B,\;\; 
\forall\, i\in[n]:\; i\in\texttt{list}_P \text{ iff } \C(P,i)\text{ is a child of } P \text{ in } \T_G\big\}$.} 
$A_i\leftarrow V_{< i}\cap\Gamma(i)$, for every $i\in [n]$;\;
$B_j\leftarrow \Gamma(j)$, for every $j\in [n]$;\;
$M_\B\leftarrow$ the $n^2\times n$ matrix whose $k$-th row is $x_k=x(P_k)$;\;
$M_G\leftarrow$ the $n\times n^2$ matrix whose $(i,j)$-th column is $x_{i,j}=x(A_i\setminus B_j)$;\; 
$M_{\B,G}\leftarrow \texttt{Rectangular\_Fast\_Matrix\_Multiplication}(M_\B, M_G)$;\;
$g^{P_k}_{ij}\leftarrow
\left\{\begin{array}{ll}
\top, & \text{ if } M_{\B, G}[k, (i,j)] > 0 \\ 
\bot, & \text{ otherwise} \\
\end{array}\right.(\text{for every } P_k\in\B \text{ and } i,j\in [n])$;\;
$L_\B\leftarrow$ a vector of lists (of integers), one $\texttt{list}_P$ for each $P\in\B$;\;
\ForEach{$P\in\B$}{
$J_P\leftarrow\{j\in [n]\mid j \text{ adjacent to all vertices in } {P}_{< j}\}$;\;
$i_{P}\leftarrow \texttt{compute-i}(P)$;\;
\ForEach{$i\in [i_P+1, n]\cap\N$ s.t. $i\not\in P$}{
$\texttt{good}\leftarrow\top$;\;
\ForEach{$j\in [1, i-1]\cap\N$}{
\lIf{$g^{P}_{ij}=\bot$ \texttt{\bf{and}} $\Big( \big(j\not\in P_{< i}\cap\Gamma(i) 
\texttt{\bf{ and }} \{j, i\}\in E \big) \texttt{\bf{or}}
 \big(j\not\in P \texttt{\bf{ and }} j\in J_P \big) \Big)$ }{
$\texttt{good}\leftarrow\bot$; 
}
}
\If{$\texttt{good}=\top$}{
$L_\B\leftarrow$ add the index $i$ to the list $\texttt{list}_P$; 
}
}
}
\Return{$L_\B$;}
}
\end{algorithm}
\smallskip

The algorithm works as follows. From line~1~to~5, the variables 
$\{A_i\}_{i=1}^n$, $\{B_j\}_{j=1}^n$ and the matrices $M_\B$, $M_G$ are initialized as they were defined in Proposition~\ref{prop:reduction}. 
Then, at line~5, the matrix $M_{\B, G}=M_\B\,M_G$ is computed by invoking a rectangular fast matrix 
multiplication algorithm, \eg the procedure devised by Le~Gall~in~\cite{Gall12}, on input $\langle M_\B, M_G\rangle$.
At this point, Algorithm~\ref{ALGO:solve-IB} is in position to compute, for each maximal clique $P\in\B$, 
all the good indices $i\in [n]$, namely, those that lead to a legitimate child $C=\C(P,i)$ of $P$. 
This is done by checking the entries of $M_{\B, G}$ corresponding to all the indices $(i,j)$ 
(as prescribed by Proposition~\ref{prop:children_char}, Lemma~\ref{lemmaA}, 
Lemma~\ref{lemmaB}, and Proposition~\ref{prop:reduction}). Whenever an index $i$ is found good for a maximal clique $P\in\B$, 
at line~15 of Algorithm~\ref{ALGO:solve-IB}, then $i$ gets added to a list at line~16, denoted $\texttt{list}_P$,  
which aims to collect all (and only) the good indices of $P$. At the end of the computation, 
the vector $L_\B$ containing all the pairs $(P, \texttt{list}_P)$ is returned at line~17.
This concludes the description of Algorithm~\ref{ALGO:solve-IB}.
Notice that, since $|\B|=n^2$ and every $P\in\B$ has at most $n$ children, then the space usage of $L_\B$ is $O(n^3)$.

The following proposition asserts the 
correctness and the time complexity of Algorithm~\ref{ALGO:solve-IB}.

\begin{Prop}\label{prop:rectangular} 
Let $G=(V,E)$ be an $n$-vertex graph. 
Consider any invocation of Algorithm~\ref{ALGO:solve-IB} on input $\B$, 
where $\B=\{P_1, \ldots, P_{n^2}\}$ is batch of $n^2$ maximal cliques of $G$. 

Then, the procedure correctly outputs a vector $L_{\B}$ such that: 
\[L_\B=\big\{(P, \texttt{list}_P) \mid P\in\B,\;\; 
\forall\, i\in[n]:\; i\in\texttt{list}_P \text{ iff } \C(P,i)\text{ is a child of } P \text{ in } \T_G\big\}.\] 
Moreover, Algorithm~\ref{ALGO:solve-IB} always halts within the following time bound: 
\[\texttt{Time}\big[\text{Algorithm~\ref{ALGO:solve-IB}}\big] = O(n^{2\omega(1,1,1/2)})=O(n^{4.093362}).\] 
In this manner, each problem instance $\I^P$ gets solved within amortized time: 
\begin{align*} 
\texttt{Time}\big[\I^{P}\big] & = O(n^{2\omega(1,1,1/2)-2})=O(n^{2.093362}), 
\end{align*} 
which is an amortized time across $n^2$ problem instances. 

Finally, the procedure employs $O(n^4)$ working space. 
\end{Prop} 
\begin{proof} 
The correctness of Algorithm~\ref{ALGO:solve-IB} follows straightforwardly from Propostion~\ref{prop:children_char}, 
Lemma~\ref{lemmaA}, Lemma~\ref{lemmaB} and Proposition~\ref{prop:reduction}. The running time has been analyzed already in the previous paragraph. 
The space bound comes from the fact that $M_{\B,G}$ has $n^4$ entries, being it of size $n^2\times n^2$. 
\end{proof}


\emph{A Remark on the Time Delay.} 
We wish to notify that it is actually possible to obtain a listing algorithm with a 
rigorous $O(n^{2\omega(1,1,1/2)-2})$ time delay: this can be done by introducing a queueing scheme,  
which firstly collects a certain polynomially bounded amount of maximal cliques of $G$ as a bootstrapping phase. 
The details of the queueing scheme are given in Subsection~\ref{subsect:strict_delay}.

%% file: clique-listing-Sect5-Backtracking.tex
\section{Batch-DFS Backtracking}\label{sect:backtracking} 
In the previous section we described how to reduce $\I^\B$ to rectangular matrix multiplication. 
Nevertheless, the description of our {\mainproblem} procedure is not yet complete:
since the algorithm needs to traverse the entire RS-tree $\T_G$ (consuming only polynomial space),
a careful backtracking procedure must be taken into account to keep the search process going on.
We propose an abstract backtracking scheme, named \emph{Batch-DFS}, 
which will make the skeleton of our {\mainproblem} solution.

\smallskip
\textbf{An Overview of Batch-DFS.}
Let $\T$ be an $n$-ary (rooted) tree of height at most $n$, for some $n\in\N$. We denote by $K_0$ the root of $\T$. 
Assume we are given a procedure $\texttt{children()}$, which takes as input a batch $\B$ of nodes of $\T$ 
and returns as output a vector $L_{\B}$ containing all the children $C$ of $P$, for every $P\in\B$. 
Notice that $0\leq |L_{\B}|\leq |\B|\,n$. Assuming we aim to visit all nodes of $\T$, 
meanwhile consuming only a polynomial amount of memory in $n$, 
then our first choice would have been to perform a DFS on $\T$. 
Nevertheless, we now show that it is possible to explore $\T$ by (somehow) 
grouping the search of new children nodes into batches $\B$, 
provided that: (1) the size of each batch is polynomially bounded in $n$, 
and (2) the backtracking phase is performed on a LIFO policy. 
The pseudocode given below encodes Batch-DFS, which will be  
our reference model of backtracking for directing the search process towards yet unexplored nodes.

\smallskip
\begin{algorithm}[H]\label{ALGO:solve}
\caption{Batch Depth-First Search.}
\DontPrintSemicolon
\nonl \SetKwProg{Fn}{Procedure}{}{}
\Fn{$\texttt{batch\_DFS}(K_0, \texttt{children(), B})$}{
\KwIn{the root $K_0$ of $\T$, where $\T$ is a tree having height $n\in\N$; 
also, a procedure $\texttt{children()}$ which takes as input a batch $\B$ of nodes of $\T$ and 
returns as output a vector $L_{\B}$ containing all the children $C$ of $P$ for every $P\in\B$;
finally, the capacity $B\in\N$, $B>0$, of any batch $\B$.}
\KwOut{a listing of all the nodes $K$ of $\T$.}
$S_{\texttt{bt}}\leftarrow \{K_0\}$;\tcp{$S_{\texttt{bt}}$ is a backtracking stack, implemented as a LIFO stack and initialized with $K_0$}
\While{$S_{\texttt{bt}}\neq\emptyset$}{
$\B\leftarrow\emptyset$;\tcp{$\B$ is a batch of nodes of $\T$ with capacity $B$, initialized to be empty}
\While{$|\B|<B$ \textbf{and} $S_{\texttt{bt}}\neq\emptyset$ }{
$P\leftarrow\texttt{pop\_from\_top}(S_{\texttt{bt}})$;\tcp{remove one single node $P$ from the top of $S_{\texttt{bt}}$}
$\B\leftarrow\B\cup\{P\}$;\;
$\texttt{print}(P)$;\tcp{print $P$ as output}
}
$L_{\B}\leftarrow \texttt{children}(\B)$;\tcp{the vector of all children $C$ of $P$, for every $P\in\B$}
$S_{\texttt{bt}}\leftarrow$ $\texttt{push all elements of $L_{\B}$ on top of $S_{\texttt{bt}}$}$;\;
}}
\end{algorithm}
\smallskip 

\textbf{Description of Batch-DFS.}
To start with, Algorithm~\ref{ALGO:solve} takes the following input: the root $K_0$ of $\T$; moreover, 
a procedure $\texttt{children()}$ (which is supposed to take in input a batch $\B$ of nodes of $\T$, 
and to return a vector $L_{\B}$ containing all children $C$ of $P$, for every $P\in\B$);
finally, a positive number $B\in\N$, representing the fixed capacity of any batch $\B$ collected at lines~4-7.
The procedure aims to provide a listing of all the nodes $K$ of $\T$.

Going into its details, Algorithm~\ref{ALGO:solve} works as follows. 
A LIFO stack $S_{\texttt{bt}}$ is maintained in order to direct the search of yet unexplored nodes.
Initially, $S_{\texttt{bt}}$ contains only the root $K_0$ of $\T$ (line~1).
Then, the procedure enters within a \texttt{while-loop}, which lasts until $S_{\texttt{bt}}\neq\emptyset$ at line~2.
Herein, the procedure tries to collect a batch $\B$ of exactly 
$|\B|=B$ nodes, picking out new nodes (as needed) from the top of the stack $S_{\texttt{bt}}$ at line~5.
Every node $P$ that is removed from $S_{\texttt{bt}}$ at line~5, and then inserted into $\B$ at line~6, 
is also printed out at line~7. Observe that even if the size of the batch fails to reach the amount $B$, 
\ie even if it happens ``$|\B|<B$ and $S_{\texttt{bt}}=\emptyset$" at line~4, then Algorithm~\ref{ALGO:solve}
moves on anyway. At line~8 the procedure $\texttt{children()}$ is invoked on input $\B$, 
in order to generating the vector $L_{\B}$ containing all children nodes $C$ of $P$, for every $P\in\B$.
Soon after, each of such child node $C$ is pushed on top of $S_{\texttt{bt}}$ (see line~9).

As already mentioned,  Algorithm~\ref{ALGO:solve} halts as soon as the condition 
``$S_{\texttt{bt}}=\emptyset$" holds at line~2.

\smallskip
\textbf{An Analysis of Batch-DFS.} 
The following propositions starts our analysis of Batch-DFS 
by showing that every node $K$ of $\T$ is eventually outputted (w/o repetitions).
\begin{Prop}\label{prop:batch_DFS}
Let $\T$ be an $n$-ary tree, having height at most $n\in \N$ and rooted in $K_0$.
Consider any invocation of Algorithm~\ref{ALGO:solve} 
on input $\langle K_0, \texttt{children()},B\rangle$, where $B\in \N$, $B>0$.
Then, every node $K$ of $\T$ is eventually outputted (at line~7), without repetitions.
\end{Prop}
\begin{proof}
\begin{itemize}

\item \emph{Fact~1.} 
We first argue that every node $K$ of $\T$ is eventually outputted. 

Let $K$ be any node of $\T$. The proof proceeds by induction on the distance $\texttt{dist}_\T(K_0, K)$ between the root $K_0$ and $K$. 
As a base case, it is easy to check (from the pseudocode of Algorithm~\ref{ALGO:solve}) 
that the root $K_0$ is printed at the first iteration of line~7. Now, assume that every node $K$ having distance at most 
$d=\texttt{dist}_\T(K_0, K)$ from $K_0$ is eventually printed out at some iteration of line~7. 
Let $\hat{K}$ be any node at distance $\texttt{dist}_\T(K_0, \hat{K})=d+1$ from $K_0$. 
Also, let $\P(\hat{K})$ be the parent of $\hat{K}$. Since $\texttt{dist}_\T(K_0, \P(\hat{K}))=d$, then at some iteration of line~7, 
$\P(\hat{K})$ is outputted, hence it is also added to $\B$ at line~6. 
Subsequently, at line~9, all children of $\P(\hat{K})$ (and thus, in particular, $\hat{K}$) 
are added on top of $S_{\texttt{bt}}$. Eventually, at some future iteration of line~5, 
$\hat{K}$ must be picked up from the top of $S_{\texttt{bt}}$. 
As that point, $\hat{K}$ must be outputted at line~7. 
Since $\hat{K}$ was arbitrary, the thesis follows.

\item \emph{Fact~2.} 
Each node $K$ can't be printed out twice. 

Indeed, when $K$ is printed at line~7, it is also removed from $S_{\texttt{bt}}$, 
and all of its successors are added on top of $S_{\texttt{bt}}$: 
this is the only way in which a node can enter within $S_{\texttt{bt}}$. 
Since $\T$ is a tree, the thesis follows.
\end{itemize}
\end{proof}

\begin{Prop}\label{prop:batch_dfs_complexity}
Let $\T$ be an $n$-ary tree, rooted at $K_0$. Consider any invocation of Algorithm~\ref{ALGO:solve} 
on input $\langle K_0, \texttt{children()},B\rangle$, where $B\in \N$, $B>0$.
In particular, let $\iota_j$ be any iteration of the \texttt{while-loop} at line~2 of Algorithm~\ref{ALGO:solve}. 
Let $\B^{(\iota_j)}$ be the corresponding batch $\B$ which is given in input to \texttt{children()} during $\iota_j$ at line~8. 
Then, the whole execution of $\iota_j$ takes time: 
\[\texttt{Time}\big[\iota_j\text{ iteration of line~2}\big] = 
	\texttt{Time}\big[\texttt{children}(\B^{(\iota_j)})\big] + O\big(n\,|\B^{(\iota_j)}|\big).
\] 
\end{Prop}
\begin{proof}
The thesis follows directly from the definition of line~9 and from the fact that $\T$ is an $n$-ary tree, so that 
the vector $L_{\B^{(\iota_j)}}$ (at line~8 of Algorithm~\ref{ALGO:solve}) contains at most $n\,|\B^{(\iota_j)}|$ elements.
\end{proof}

In the next proposition we argue that $S_{\texttt{bt}}$ can grow up its size at most polynomially in $n$, $B$. 
Before proving that, we shall introduce some notation.

Let us consider any two consecutive iterations of the \texttt{while-loop} at line~2 of Algorithm~\ref{ALGO:solve}, 
say the iterations $\iota_j$ and $\iota_{j+1}$. For any $j\geq 1$, let $\B^{(\iota_j)}$ be the batch $\B$ which is given in input 
to $\texttt{children()}$ at line~8 and during $\iota_j$; 
moreover, let $L_{\B^{(\iota_j)}}$ be the vector of nodes returned 
by the invocation of $\texttt{children}(\B^{(\iota_j)})$ at line~8 during $\iota_j$. 
We shall say that Algorithm~\ref{ALGO:solve} \emph{backtracks} at the $\iota_{j+1}$ iteration whenever it holds that:
$\B^{(\iota_{j+1})}\not\subseteq L_{\B^{(\iota_j)}}$, 
\ie whenever, at the $\iota_{j+1}$ iteration of line~8, the batch $\B^{(\iota_{j+1})}$ contains some nodes 
that were not pushed on $S_{\texttt{bt}}$ at the $\iota_j$ iteration of lines~$9\sim11$, 
but during some previous iteration $\iota_{k}$ (for some $k<j$) instead.

\begin{Prop}\label{prop:batch1}
Let $\T$ be an $n$-ary tree having root $K_0$ and total height at most $n\in \N$.
Consider any invocation of Algorithm~\ref{ALGO:solve} on input $\langle K_0, \texttt{children()},B\rangle$.
Throughout the whole execution, the backtracking stack $S_{\texttt{bt}}$ can grow up to contain at most $n^2B$ nodes. 
For this reason, Algorithm~\ref{ALGO:solve} consumes at most $O\big(n^2B + \texttt{Space}\big[\texttt{children()}\big]\big)$ space,  
where $\texttt{Space}[\texttt{children()}]$ denotes the worst-case space consumed by any invocation of \texttt{children()}. 
\end{Prop}
\begin{proof}
Since $\T$ is $n$-ary and $|\B|\leq B$ (because of line~4 of Algorithm~\ref{ALGO:solve}), 
then each batch $\B$ has at most $nB$ children. Since $S_{\texttt{bt}}$ is accessed adopting a LIFO policy, 
and since $\T$ has total height at most $n$, the following fact holds: 
until the first backtrack doesn't happen, $S_{\texttt{bt}}$ can grow its size up to $nB$ elements at most $n$ times. 
Therefore, $S_{\texttt{bt}}$ can grow its size up to $n^2 B$ elements, 
before it needs to backtrack at some iteration of the \texttt{while-loop} at line~2. 
As soon as Algorithm~\ref{ALGO:solve} starts to backtrack, say at the $\iota_{j+1}$ iteration, 
then $S_{\texttt{bt}}$ shrinks its size, collecting (at lines~$4\sim7$) a batch $\B^{(\iota_{j+1})}$ that must contain  
some nodes which had been pushed on $S_{\texttt{bt}}$ at some previous iteration $\iota_{k}$ of lines~$9\sim11$ 
(for some $k<j$). We now observe that, at the $\iota_{j+1}$ iteration of line~8, 
the stack $S_{\texttt{bt}}$ must contain at most as many elements as it contained 
at the end of the $\iota_{k}$ iteration. For this reason, $S_{\texttt{bt}}$ 
has still no way to grow its size up to more than $n^2 B$, 
by going down the levels of $\T$ once again after that a backtracking occurred. The same observation 
continues to hold for any possible subsequent backtracking. In this manner $S_{\texttt{bt}}$ can grow 
its size up to $n^2B$ nodes at most during the whole computation. 
\end{proof}

We now aim to show another crucial property of Algorithm~\ref{ALGO:solve}, 
one that turns out decisive for adopting Batch-DFS in order to speed-up the {\mainproblem} problem. 

In order to prove this fact, it is convenient to introduce a three-way coloring scheme on $\T$. 
\paragraph{A Three-Way Coloring on $\T$.} 
Consider any invocation of Algorithm~\ref{ALGO:solve} on the following 
input $\langle K_0, \texttt{children()},B\rangle$, where $K_0$ is the root of $\T$.
At the beginning of the execution, it is prescribed that all nodes $K$ of $\T$ are colored in white.
As soon as a white node $K$ of $\T$ is pushed on top of $S_{\texttt{bt}}$ 
(either at line~1 or at line~11 of Algorithm~\ref{ALGO:solve}), then $K$ changes its colour from white 
to green. Stated otherwise, at each step of Algorithm~\ref{ALGO:solve}, all the nodes in $S_{\texttt{bt}}$ are green. Finally, as soon as any $K$ gets removed from $S_{\texttt{bt}}$ at line~5, 
then $K$ changes its colour from green to black. 

Observe that, since by Proposition~\ref{prop:batch1}
every node of $\T$ is eventually pushed on $S_{\texttt{bt}}$, 
and then removed from it, exactly once, 
then every node of $\T$ eventually transits from white to green, and then from green to black.
Moreover, black nodes remain such until the end of the execution.

We proceed by observing an invariant which is maintained by Algorithm~\ref{ALGO:solve}.
\begin{Lem}\label{lem:batch2} 
Let $\T$ be an $n$-ary tree having root $K_0$ and total height at most $n\in \N$. 
Consider any invocation of Algorithm~\ref{ALGO:solve} on input $\langle K_0, \texttt{children()},B\rangle$, 
and let $\sigma_i$ be any step of execution of line~3. 
Let us denote by $\ell^{\sigma_i}_{\texttt{green}}$ the minimum distance  
between the root $K_0$ and any node of $\T$ which is green at step $\sigma_i$, 
\ie 
\[\ell^{\sigma_i}_{\texttt{green}} = 
\min\big\{\texttt{dist}_{\T}(K_0, K)\mid 
K\in\T \text{ and } K\text{ is green at execution step } \sigma_i\big\}.
\]
Then, at step $\sigma_i$ every node $K\in\T$ such that 
$\texttt{dist}_{\T}(K_0, K)\leq\ell^{\sigma_i}_{\texttt{green}}$ is either 
green or black but it is not white.
\end{Lem}
\begin{proof}
Assume, for the sake of contradiction, that at step $\sigma_i$ 
there exists a white node $K_w$ in $\T$ such that $\texttt{dist}_\T(K_0, K_w)\leq\ell^{\sigma_i}_{\texttt{green}}$. 
Recall that, at the beginning of the execution, the root $K_0$ of $\T$ turns green at line~1;
hence, at any subsequent step, $K_0$ must be either green or black. 
Thus, at $\sigma_i$, there must exist at least one ancestor of $K_w$ which 
is either green or black but not white, because there is a path from $K_w$ to $K_0$. 
Now, let $\hat{K}$ be the ancestor of $K_w$ which is either green or black and such that its distance 
from $K_0$ is maximum among all of those ancestors of $K_w$ that are either green or black.
What is the colour of $\hat{K}$ at step $\sigma_i$, is it green or is it black? 

Notice that $\hat{K}$ is not green at $\sigma_i$; in fact, since $\hat{K}$ is an ancestor of $K_w$, 
then: \[\texttt{dist}_\T(K_0, \hat{K}) < \texttt{dist}_\T(K_0, K_w)\leq \ell^{\sigma_i}_{\texttt{green}},\] 
whereby a green colored $\hat{K}$ would 
contradict the minimality of $\ell^{\sigma_i}_{\texttt{green}}$. 
Still, $\hat{K}$ is not even black at $\sigma_i$; 
otherwise, all the children of $\hat{K}$ would have been colored in green at some previous step of the algorithm, because of lines~$8\sim11$ of Algorithm~\ref{ALGO:solve}, 
thus contradicting the maximality of $\texttt{dist}_\T(K_0, \hat{K})$.
No colour is actually possible for $\hat{K}$ at $\sigma_i$, this leads to a contradiction. 

Indeed, there exists no such a white node $K_w$. This implies the thesis.
\end{proof}
\begin{Prop}\label{prop:batch2}
Let $\T$ be an $n$-ary tree having root $K_0$ and total height at most $n\in \N$.
Consider any invocation of Algorithm~\ref{ALGO:solve} on input $\langle K_0, \texttt{children()},B\rangle$.
Then, the total number of steps of execution of 
line~4 in which the condition ``$S_{\texttt{bt}}=\emptyset$" 
holds is always less than or equal to $n$.
\end{Prop}
\begin{proof} 
Let's consider any generic step of execution of line~4, say step $\sigma_j$, such that ``$S_{\texttt{bt}}=\emptyset$" holds.
Let $\sigma_i$, for some $i<j$, be the last step of execution of line~3 which precedes $\sigma_j$. 
Stated otherwise, we are considering a sequence of execution steps,  
$\sigma_i, \sigma_{\texttt{next\_step}(i)}, \ldots, \sigma_j$,  
where:
\begin{itemize}
\item the starting step $\sigma_i$ corresponds to an execution of line~3;
\item $\sigma_{\texttt{next\_step}(i)}, \ldots$ marks the 
(immediately following) entrance of the computation process into the \texttt{while-loop} at line~4;
\item $\sigma_j$ corresponds to the (subsequent) exhaustion of the \texttt{while-loop} at line~4; 
\ie $\sigma_j$ is the first step of execution of line~4, subsequent to $\sigma_i$,  
such that the condition ``$S_{\texttt{bt}}=\emptyset$" holds.
\end{itemize}
By Lemma~\ref{lem:batch2}, at step $\sigma_i$, every node $K$ such that 
$\texttt{dist}_{\T}(K_0, K)\leq\ell^{\sigma_i}_{\texttt{green}}$ must be either green or black, 
and notice that there must exist at least one such 
green node at the $\sigma_i$ step of execution of line~3 
(otherwise we would have had $S_{\texttt{bt}}=\emptyset$ just before at line~2).
Since, at step $\sigma_j$, $S_{\texttt{bt}}=\emptyset$ holds by hypothesis, then
every node $K$ such that $\texttt{dist}_{\T}(K_0, K)\leq\ell^{\sigma_i}_{\texttt{green}}$, 
must be turned black at $\sigma_j$.
Stated otherwiese, all nodes having distance $\ell^{\sigma_i}_{\texttt{green}}$ from $K_0$ 
that were green at step $\sigma_i$ must be turned black at step $\sigma_j$. 
In this manner, we see that at step $\sigma_j$ yet another level of depth 
in $\T$ has been loosely speaking ‘‘turned-off’’ completely and forever.

Since $\T$ has total height at most $n$, the thesis follows.
\end{proof}

%% file: clique-listing-Sect6-Algorithm.tex
\section{An Asymptotically Faster Algorithm for {\mainproblem}}\label{sect:algorithm}

The present section offers two algorithms for {\mainproblem}. Our core procedure is Algorithm~\ref{ALGO:main}: 
it provides a listing of all the maximal cliques of any given $n$-vertex graph with a time delay polynomial in $n$. 
However, due to technical reasons (related to Proposition~\ref{prop:rectangular}), 
the procedure exhibit a time delay that is in some sense ``amortized" across $n^2$ output operations.
In subsection~\ref{subsect:strict_delay}, Algorithm~\ref{ALGO:strict} will be introduced in order to overcome this issue, thus 
achieving the time delay stated in Theorem~\ref{thm:strict}.

The pseudocode of Algorithm~\ref{ALGO:main} is presented here below.

\smallskip
\begin{algorithm}[H]
\caption{Listing all Maximal Cliques.}\label{ALGO:main}
\DontPrintSemicolon
\nonl \SetKwProg{Fn}{Procedure}{}{}
\Fn{$\texttt{list\_\texttt{MC}}(G)$}{
\KwIn{A graph $G=(V,E)$, where $|V|=n$.}
\KwOut{A listing of all the maximal cliques $K$ of $G$.}
$K_0\leftarrow$ construct the lexicographically greatest maximal clique $K_0$;\;
$\texttt{batch\_DFS}(K_0, \texttt{children()}, n^2)$;\tcp{invoke Algorithm~\ref{ALGO:solve}}
}

\SetKwProg{SubFn}{SubProcedure}{}{}
\nonl\SubFn{$\texttt{children}(\B)$}{
\setcounter{AlgoLine}{0}
\KwIn{A (non-empty) batch $\B=\{P_1, \ldots, P_{n^2}\}$ of $|\B|\leq n^2$ maximal cliques of $G$.}
\KwOut{The vector $\B'$ of all the children of $\B$, 
\ie $\B'=\{(P, \texttt{list}_P) \mid P\in\B,\;\; 
\forall\, i\in[n]:\; i\in\texttt{list}_P \text{ iff } 
	\C(P,i)\text{ is a child of } P \text{ in } \T_G\}$.}
\If{$|\B|=n^2$}{
$\B'\leftarrow\texttt{Solve\_Rectangular\_$\I$}(\B)$;\tcp{invoke Algorithm~\ref{ALGO:solve-IB}}
}
\Else{
$\B'\leftarrow$ compute all children of $\B$ with Makino-Uno's procedure~\cite{MU04};\;}
\Return{$\B'$;}
}
\end{algorithm}
\smallskip

\textbf{Description of Algorithm~\ref{ALGO:main}.} 
At line~1, the lexicographically greatest maximal clique $K_0$ of $G$ gets constructed. At line~2, Algorithm~\ref{ALGO:solve} 
is invoked on input $\langle K_0, \texttt{children()}, n^2\rangle$.
The subprocedure \texttt{children()} is defined as follows. 
It takes in input a (non-empty) batch $\B$ containing $|\B|\leq n^2$ maximal cliques of $G$, 
and it aims to return a vector $L_\B$ containing all (and only) the children of $\B$, 
\ie $L_{\B}=\{(P, \texttt{list}_P) \mid P\in\B,\;\; 
\forall\, i\in[n]:\; i\in\texttt{list}_P \text{ iff } \C(P,i)\text{ is a child of } P \text{ in } \T_G\}$.
Given $\B$ in input, the course of actions within \texttt{children()} dependes on the size $|\B|$:
\begin{itemize}
\item if $|\B|=n^2$, then $L_{\B}$ is computed by invoking Algorithm~\ref{ALGO:solve-IB} on input $\B$ at line~2;
\item otherwise, if $0<|\B|<n^2$, then $L_\B$ is computed 
      with the original algorithm of Makino and Uno~\cite{MU04} 
     (the one having an $O(n^{\omega})$ time delay complexity). 
\end{itemize} 
Then, $L_{\B}$ is returned as output at line~5 of \texttt{children()}.
There's still one missing detail. Recall the functioning of Algorithm~\ref{ALGO:solve}:
at line~5, \texttt{pop\_from\_top()} is assumed to retrieve one single maximal clique from $S_{\texttt{bt}}$ (and not a pair $(P, \texttt{list}_P$).
For this reason, a careful implementation of \texttt{pop\_from\_top()} must be taken into account.
It may go as follows. Firstly, \texttt{pop\_from\_top()} accesses to the head of the stack $S_{\texttt{bt}}$, say $(P, \texttt{list}_P)$, 
without actually removing it from $S_{\texttt{bt}}$. There, it removes the first element of $\texttt{list}_P$, say $\hat{i}$,    
thus reducing the size of $\texttt{list}_P$ by one unit. 
At this point, $(P, \texttt{list}_P)$ is removed from the top of $S_{\texttt{bt}}$ 
if and only if $\texttt{list}_P$ has become empty by removing $\hat{i}$. 
Finally, \texttt{pop\_from\_top()} constructs the maximal clique $\C(P,\hat{i})$, 
by invoking Algorithm~\ref{ALGO:compute-C} on input $(P_{< \hat{i}}\cap\Gamma({\hat{i}}))\cup\{{\hat{i}}\}$. 
Notice that any invocation of \texttt{pop\_from\_top()} takes time $O(n^2)$, which is due to Algorithm~\ref{ALGO:compute-C}. 
This concludes the description of \texttt{pop\_from\_top()}, and thus that of Algorithm~\ref{ALGO:main}.
The following proposition asserts its correctness.

\begin{Prop}\label{prop:prealgocorrectness}
On input $G=(V,E)$, the procedure Algorithm~\ref{ALGO:main} provides a listing of all the maximal 
cliques of $G$ without repetitions. 
\end{Prop}
\begin{proof} Observe that Algorithm~\ref{ALGO:main} invokes Algorithm~\ref{ALGO:solve} at line~2. 
Also, Proposition~\ref{prop:rectangular} implies that the subprocedure $\texttt{children()}$ 
of Algorithm~\ref{ALGO:main} is correct. The thesis follows by Propostion~\ref{prop:batch_DFS}.
\end{proof}

The following proposition asserts the time complexity of Algorithm~\ref{ALGO:main}. 

\begin{Prop}\label{prop:prealgo_timecomplexity}
Given any $n$-vertex graph $G=(V,E)$ as input, Algorithm~\ref{ALGO:main} outputs  
the first $x$ maximal cliques of $G$ within the following time bound, for any $x\in\N$:
\[\tau_{\texttt{first\_}x}= O\big(n^{\omega+3}+xn^{2\omega(1,1,1/2)-2}\big) = O\big(n^{5.3728639} + xn^{2.093362}\big).\]
\end{Prop}
\begin{proof} The proof is divided into four steps. There, 
$\iota_j$ will denote any generic (but fixed) iteration of the \texttt{while-loop} at line~2 of Algorithm~\ref{ALGO:solve}. 
\begin{enumerate}
\item[\emph{Fact~1.}] There exist at most $n$ iterations $\iota_j$ of the \texttt{while-loop} 
at line~2 of Algorithm~\ref{ALGO:solve} such that $0<|\B^{(\iota_j)}|<n^2$; 
and for all other iterations $\iota_j$ of line~2 of Algorithm~\ref{ALGO:solve}, it holds $|\B^{(\iota_j)}|=n^2$.

\emph{Proof of Fact~1.} 
Since $B=n^2$, we have that $0<|\B^{(\iota_j)}|<n^2$ holds if and only if the condition 
``$S_{\texttt{bt}}=\emptyset$" holds at line~4 of Algorithm~\ref{ALGO:solve} during $\iota_j$. 
By Proposition~\ref{prop:batch2}, this may happen at most $n$ times throughout the whole execution of Algorithm~\ref{ALGO:solve}.

\item[\emph{Fact~2.}] All maximal cliques of $\B^{(\iota_j)}$ are outputted with $O(n^2)$ time delay.

\emph{Proof of Fact~2.} Notice that during $\iota_j$ all maximal cliques in $\B^{(\iota_j)}$ are outputted at line~7 of 
Algorithm~\ref{ALGO:solve}. Just before, at line~5, \texttt{pop\_from\_top()} needs to make an invocation to Algorithm~\ref{ALGO:compute-C} 
(as we had already observed in the description of Algorithm~\ref{ALGO:main}.)
By Propostion~\ref{compute-C}, this latter invocation takes at most $O(n^2)$ time. 

\item[\emph{Fact~3.}] 
If $|\B^{(\iota_j)}|=n^2$, then the whole execution of $\iota_j$ takes time 
$O\big(n^{2\omega(1,1,1/2)}\big)$.

\emph{Proof of Fact~3.} This follows by Proposition~\ref{prop:batch_dfs_complexity} and  
Proposition~\ref{prop:rectangular}.

\item[\emph{Fact~4.}] If $|\B^{(\iota_j)}|<n^2$, then the whole execution of $\iota_j$ takes time  
$ O\big(n^{\omega}|\B^{(\iota_j)}| \big) = O\big(n^{\omega+2}\big)$.

\emph{Proof of Fact~4.}
This follows from Proposition~\ref{prop:batch_dfs_complexity} 
and from existence of the $O(n^{\omega})$ procedure devised by Makino and Uno in~\cite{MU04}. 
\end{enumerate}
By Facts~$1\sim 4$, the above mentioned time bound on $\tau_{\texttt{first}\_x}$ follows.
\end{proof}

To conclude, the space usage of Algorithm~\ref{ALGO:main} is analyzed below.
\begin{Prop}\label{prop:prealgo_space}
The space usage of Algorithm~\ref{ALGO:main} is $O(n^4)$.
\end{Prop}
\begin{proof}
By Proposition~\ref{prop:batch1}, the space usage of Algorithm~\ref{ALGO:solve} is 
$O\big(n^2B + \texttt{Space}\big[\texttt{children()}\big]\big)$,
where $\texttt{Space}[\texttt{children()}]$ is the worst-case space consumed by any invocation of \texttt{children()}.
There is still one detail that it is worth stating. Even though to represent a maximal clique requires $O(n)$ space in memory,
recall that within the backtracking stack $S_{\texttt{bt}}$ of Algorithm~\ref{ALGO:solve} 
we have choosen to represent all the children of any generic maximal clique $K$
by keeping in memory the pair $(P, \texttt{list}_P)$, where $\texttt{list}_P$ is a list of integers having length at most $n$.
This fact implies that, in order to store all the $O(nB)$ children of any batch $\B$ of $B$ maximal cliques, we need only $O(nB)$ space.
For this reason, the stack $S_{\texttt{bt}}$ consumes only $O(n^2B)$ space throughout the whole execution of Algorithm~\ref{ALGO:solve}, 
as shown by Proposition~\ref{prop:batch1} in the abstract setting. 
Now, concerning the {\mainproblem} problem, we have $B=n^2$ (see line~2 of Algorithm~\ref{ALGO:main}).
Moreover, by Proposition~\ref{prop:rectangular}, the following holds: 
$ \texttt{Space}\big[\texttt{children()}\big] \leq \texttt{Space}\big[Algorithm~\ref{ALGO:solve-IB}\big]=O(n^4)$.
These facts imply that the space usage of Algorithm~\ref{ALGO:main} is $O(n^4)$. 
\end{proof}

Theorem~\ref{thm:firstx} follows, at this point, from Proposition~\ref{prop:prealgocorrectness}, 
\ref{prop:prealgo_timecomplexity}, and \ref{prop:prealgo_space}.

\subsection{An $O(n^{2\omega(1,1,1/2)-2})$ Time Delay Algorithm for {\mainproblem}: 
 \\ The Proof of Theorem~\ref{thm:strict}}\label{subsect:strict_delay}

This subsection describes Algorithm~\ref{ALGO:strict}, which is the procedure mentioned in Theorem~\ref{thm:strict}. 

The corresponding pseudocode follows below.

\smallskip
\begin{algorithm}[H]
\caption{Listing all Maximal Cliques as in Theorem~\ref{thm:strict}.}\label{ALGO:strict}
\DontPrintSemicolon
\nonl \SetKwProg{Fn}{Procedure}{}{}
\Fn{$\texttt{solve\_\mainproblem}(G)$}{
\KwIn{A graph $G=(V,E)$, where $|V|=n$.}
\KwOut{A listing of all the maximal cliques $K$ of $G$.}
$\tau_{\texttt{delay}}\leftarrow c_0 \lceil n^{2\omega(1,1,1/2)-2}\rceil=c_0\lceil n^{2.093362}\rceil$; \tcp{for some sufficiently large constant $c_0>0$}
$T\leftarrow c_1 \lceil n^{\omega - 2\omega(1,1,1/2) + 5}\rceil= c_1 \lceil n^{3.2795019}\rceil $; \tcp{for some sufficiently large constant $c_1>0$}
$Q\leftarrow\emptyset$;\tcp{let $Q$ be an empty queue}
$i\leftarrow \texttt{boot}(Q,T,G)$;\tcp{the bootstrap aims to fill $Q$ up to containing $T$ elements}
$\texttt{counter}\leftarrow 0$;\;
\While{$\sigma_i\neq \sigma_{\texttt{end}}$ }{ \tcp{\ie while $\sigma_i$ is not 
the last step of \texttt{list\_\texttt{MC}($G$)}'s computation.}
	$\sigma_{i+1} \leftarrow \texttt{next\_step}(\texttt{list\_MC}(G), \sigma_i)$;\;
	$\texttt{counter}\leftarrow \texttt{counter}+1$;\;
	\If{$\sigma_{i+1} = \texttt{print(K)}$ at line~7 of Algorithm~\ref{ALGO:solve} }{
		$Q\leftarrow \texttt{append\_on\_tail}(Q, K)$; \tcp{ do not perform 
			the actual printing, but append $K$ to the tail of $Q$ instead}
	}
	\If{ ( $|Q|>0$ \textbf{ and } $\texttt{counter}\geq\tau_{\texttt{delay}}$ ) \textbf{ or } $|Q| > T+n^2$}{
		$K\leftarrow \texttt{remove\_from\_head}(Q)$;\tcp{remove the head $K$ of $Q$}
		$\texttt{print}(K)$;\tcp{perform the actual printing of $K$}
		$\texttt{counter}\leftarrow 0$;\;
	}
	$i\leftarrow i+1$;\;
}
\While{$|Q|>0$}{
	$K\leftarrow \texttt{remove\_from\_head}(Q)$;\tcp{remove the head $K$ of $Q$}
	$\texttt{print}(K)$;\tcp{perform the actual printing of $K$}
}
}
\end{algorithm}
\smallskip

Algorithm~\ref{ALGO:strict} takes  as input an $n$-vertex graph $G=(V,E)$, and provides a listing of all the maximal cliques $K$ of $G$.
An overview of the algorithm follows. 
As a Turing Machine can be programmed in order to simulate each step of the computation of any another Turing Machine, 
Algorithm~\ref{ALGO:strict} performs a step-by-step simulation of the computation performed by Algorithm~\ref{ALGO:main} on input $G$. 
Given a generic step of such a computation, say $\sigma_i$, we shall denote by $\sigma_{i+1}$ the next step 
within the sequence of all steps of the computation. In particular, we shall adopt the notation 
$\sigma_{i+1} \leftarrow \texttt{next\_step}(\texttt{list\_MC}(G), \sigma_i)$. 
Stated otherwise, we are assuming that any invocation of Algorithm~\ref{ALGO:main} on input $G$ 
leads to the following sequence of steps of computation:
\[
\big\langle \sigma_0,\;\; 
\sigma_1=\texttt{next\_step}(\texttt{list\_MC}(G), \sigma_0),\;\;  
\sigma_2=\texttt{next\_step}(\texttt{list\_MC}(G), \sigma_1),\;\; \ldots\;\;,\;\; \sigma_{\texttt{end}} \big\rangle 
\]
where each $\sigma_i$ represents the execution of a particular line within the corresponding reference pseudocode.
The rationale of this being that, at each one of those steps of execution $\sigma_i$, Algorithm~\ref{ALGO:strict} 
assesses how to best manage a queue $Q$ whose aim is to collect a suitable number of maximal cliques of $G$ 
in order to sustain the time delay to scheme. 

At each $\sigma_i$, the course of actions taken by Algorithm~\ref{ALGO:strict} on $Q$ depends on: 
\begin{enumerate}
\item the current size of $Q$, \ie the number of maximal cliques that are inside $Q$ at step $\sigma_i$;
\item the numeric value of the current step-counter $i$ reached by $\sigma_i$;
\item the particular line\footnote{\ie the particular line within the pseudocode of Algorithm~\ref{ALGO:main} and Algorithm~\ref{ALGO:solve}.} 
of Algorithm~\ref{ALGO:strict} that is executed at step $\sigma_i$;
\end{enumerate}
In particular, every \texttt{print(K)} operation performed by Algorithm~\ref{ALGO:main} is hooked by Algorithm~\ref{ALGO:strict},  
where the idea there is that of appending $K$ to the tail of $Q$ without printing it (immediately) as output, 
but with the intention to perform the actual printing operation later on, 
in such a way as to keep the time delay under $O(n^{2\omega(1,1,1/2)-2})$.

Going into the details, Algorithm~\ref{ALGO:strict} is organized into three phases:
(1) \emph{initialization}, (2) \emph{bootstrapping}, and (3) \emph{listing}. These are described next.
\begin{enumerate}
\item \textbf{Initialization Phase.} To start with, some variables gets initialized.
At line~1, $\tau_{\texttt{delay}}=c_0\lceil n^{2\omega(1,1,1/2)-2}\rceil=c_0\lceil n^{2.093362} \rceil$ 
marks the time delay that the procedure aims to sustain. Here, $c_0$ is some sufficiently large absolute constant 
(whose magnitude will be clarified in the proof of Proposition~\ref{prop:finaldelay}).
At line~2, $T=c_1 \lceil n^{\omega - 2\omega(1,1,1/2) + 5}\rceil=c_1 \lceil n^{3.2795019}\rceil$ 
is the number of maximal cliques that the bootstrapping phase will aim to collect (the magnitude of $c_1$ will be also 
clarified in the proof of Proposition~\ref{prop:finaldelay}).
Finally, at line~3, the queue $Q$ is initialized to be empty.

\smallskip
\begin{algorithm}[H]
\caption{The Bootstrapping Phase.}\label{ALGO:boot}
\DontPrintSemicolon
\SetKwProg{SubFn}{SubProcedure}{}{}
\nonl\SubFn{$\texttt{boot}(Q, T, G)$}{
\setcounter{AlgoLine}{0}
\KwIn{A reference to $Q$, a threshold $T$ on the size of $Q$, the input graph $G$.}
\KwOut{the index of the current computation step $\sigma_i$, that is reached after the bootstrap.}
$\sigma_0 \leftarrow $ the starting step of Algorithm~\ref{ALGO:main}'s computation sequence on input $G$.\;
$i\leftarrow 0$;\;
\While{$|Q| < T$ \textbf{ or } $\sigma_i\neq $ line~2 of Algorithm~\ref{ALGO:solve} }{
	$\sigma_{i+1} \leftarrow \texttt{next\_step}(\texttt{list\_MC}(G), \sigma_i)$;\;
	\If{$\sigma_{i+1} = \texttt{print(K)}$ at line~7 of Algorithm~\ref{ALGO:solve} }{
		$Q\leftarrow \texttt{append\_on\_tail}(Q, K)$; \tcp{don't execute 
			\texttt{print($K$)}, append $K$ to $Q$ instead}
	}
	$i\leftarrow i+1$;\;
}
\Return{i};
}
\end{algorithm}
\smallskip

\item \textbf{Bootstrapping Phase.} This phase begins (at line~4 of Algorithm~\ref{ALGO:strict}) 
by invoking Algorithm~\ref{ALGO:boot} on input $\langle Q,T,G\rangle$. The objective is to collect 
at least $T$ maximal cliques inside $Q$. For this reason, Algorithm~\ref{ALGO:boot} starts a step-by-step simulation of 
Algorithm~\ref{ALGO:main} on input $G$. The simulation starts, at line~1, by considering the first step $\sigma_0$  of the computation.
The subsequent steps of the computation are simulated within the \texttt{while-loop} defined at line~3 of Algorithm~\ref{ALGO:boot}, 
by invoking \texttt{next\_step()} at each iteration of line~4. Whenever Algorithm~\ref{ALGO:main} performs a \texttt{print($K$)} operation 
of some maximal clique $K$ (which is checked at line~5 of Algorithm~\ref{ALGO:boot}), 
then $K$ is appended to the tail of $Q$ at line~6 (without performing the actual printing operation).
After that $Q$ gets to contain at least $T$ elements, Algorithm~\ref{ALGO:boot} extends the simulation of 
Algorithm~\ref{ALGO:main} still for awhile. 
In particular, the simulation is extended until the end of the current iteration of the \texttt{while-loop} 
at line~2 of Algorithm~\ref{ALGO:solve} which it is being simulated (for this reason, the condition ``$\sigma_i\neq\text{line~2}$ 
of Algorithm~\ref{ALGO:solve}" is checked at line~3 of Algorithm~\ref{ALGO:boot}). 
Finally, Algorithm~\ref{ALGO:boot} halts at line~8, by returning the current step counter $i$ that has been reached so far. 

\item \textbf{Listing Phase.} The listing phase begins soon after, 
at line~5 of Algorithm~\ref{ALGO:strict}, where a \texttt{counter} variable is initialized. 
Then, Algorithm~\ref{ALGO:strict} enters within the \texttt{while-loop} at line~6, whose purpose 
is that of continuing the same simulation of Algorithm~\ref{ALGO:main} that Algorithm~\ref{ALGO:boot} had begun by bootstrapping.
This time the simulation process will continue until the end, \ie until last step $\sigma_{\texttt{end}}$ of Algorithm~\ref{ALGO:main}. 
For this reason, the condition ``$\sigma_i\neq\sigma_{\texttt{end}}$" is checked at each iteration of line~6.
Observe, that each step $\sigma_i$ gets iterated to $\sigma_{i+1}$ at step~7, where \texttt{next\_step()} is invoked.
Soon after, the \texttt{counter} variable is incremented at line~8, and then the current execution step $\sigma_{i+1}$ is inspected at line~9:
if $\sigma_{i+1}$ consists into a \texttt{print($K$)} operation 
(which may have been executed only at line~7 of Algorithm~\ref{ALGO:solve}), 
then the maximal clique $K$ is appended to the tail of $Q$ at line~10, and the actual printing operation is postponed.
At line~11 the procedure checks whether it is time to execute an ouput printing,
and this happens if and only if any one of the following two conditions is met:
\begin{itemize}
\item $Q$ is not empty and the simulation of Algorithm~\ref{ALGO:main} performed more than $\tau_{\texttt{delay}}$ steps since  
the last time that a printing operation was executed at line~13 of Algorithm~\ref{ALGO:strict} 
(to verify this, the condition ``$|Q|>0$ and $\texttt{counter}\geq\tau_{\texttt{delay}}$" is checked at line~11).
\item $Q$ contains more than $T+n^2$ elements (for this reason, the condition ``$|Q|>T+n^2$" is checked at line~11 as well).
\end{itemize}
If one of the above conditions is met, then a maximal clique $K$ is removed from the head of $Q$ at line~12, 
and it is outputted by executing \texttt{print($K$)} at line~13. 
In this case, the \texttt{counter} variable is also reset to zero  at line~14.
At line~15, the step counter $i$ gets incremented (so that to prepare the ground for the next step of the simulation). 
When the \texttt{while-loop} at line~6 is completed (\ie when the simulation of Algorithm~\ref{ALGO:main} reaches $\sigma_{\texttt{end}}$) 
then every maximal clique that is still inside $Q$ is removed from it at line~17 and outputted soon after at line~18 
(for this reason, the condition ``$|Q|>0$" is checked at line~16 of Algorithm~\ref{ALGO:strict}).
\end{enumerate}
This concludes the description of Algorithm~\ref{ALGO:strict}.

We are now in position to prove Theorem~1. We shall go through a sequence of propositions. 
\begin{Prop} 
On input $G = (V, E)$, the procedure Algorithm~\ref{ALGO:strict} provides a listing of all the
maximal cliques of $G$ without repetitions.
\end{Prop}
\begin{proof}
Recall that Algorithm~\ref{ALGO:strict} performs a simulation of Algorithm~\ref{ALGO:main},   
and that it hooks all of the corresponding output printing operations.
Also recall that, by Proposition~\ref{prop:prealgocorrectness}, Algorithm~\ref{ALGO:main} outputs every maximal clique of $G$ exactly once.
This implies that every maximal clique of $G$ must enter within the queue $Q$ exactly once, 
either at line~6 of Algorithm~\ref{ALGO:boot} or at line~10 of Algorithm~\ref{ALGO:strict}.
With this in mind, from lines~12-13 and lines~17-18 of Algorithm~\ref{ALGO:strict} it follows that 
whenever a maximal clique $K$ is removed from $Q$, then $K$ is also printed out. 
Notice that, at lines~16-18 of Algorithm~\ref{ALGO:strict}, $Q$ is emptied anyhow. 
These facts imply the thesis.
\end{proof}

\begin{Prop}\label{prop:finalcorrect}
Algorithm~\ref{ALGO:boot} (\ie the bootstrapping phase of Algorithm~\ref{ALGO:strict}) always halts within time:  
$\tau_{\texttt{boot}} = O\big(n^{\omega+3}\big)=O\big(n^{5.3728639}\big)$.
\end{Prop}
\begin{proof} Recall that Algorithm~\ref{ALGO:boot} keeps the simulation of Algorithm~\ref{ALGO:main} going 
until the queue $Q$ doesn't get to contain at least $T=O(n^{\omega-2\omega(1,1,1/2)+5}) = O(n^{3.2795019})$ elements.

By Proposition~\ref{prop:prealgo_timecomplexity}, Algorithm~\ref{ALGO:main} collects $T$ elements within the following time bound: 
\begin{align*} \tau_{\texttt{boot}} & = O(n^{\omega+3} + Tn^{2\omega(1,1,1/2)-2}) =  
			O( n^{\omega+3} + n^{\omega-2\omega(1,1,1/2)+5} n^{2\omega(1,1,1/2)-2} )   \\ 
  				    & = O(n^{\omega+3}) = O\big(n^{5.3728639}\big) 
\end{align*}
Thus, Algorithm~\ref{ALGO:boot} also halts within time  
$\tau_{\texttt{boot}}=O(n^{\omega+3})=O\big(n^{5.3728639}\big)$.
\end{proof}

\begin{Prop}\label{prop:finaldelay}
The time delay between the outputting of any two consecutive maximal cliques in Algorithm~\ref{ALGO:strict} is:
$\tau_{\texttt{delay}}= O\big(n^{2\omega(1,1,1/2)-2}\big) = O\big(n^{2.093362}\big)$.
\end{Prop}
\begin{proof} Observe that every printing operation performed by Algorithm~\ref{ALGO:strict} is executed either
at line~13 or at line~18. The time delay between any two consecutive iterations of line~18 is only $O(1)$.
Thus, we shall focus on proving the thesis with respect to line~13.
Let's recall the functioning of Algorithm~\ref{ALGO:strict} and that of Algorithm~\ref{ALGO:solve}.
Consider any generic iteration of the \texttt{while-loop} at line~2 of Algorithm~\ref{ALGO:solve}, say the $\iota_j$ iteration. 
Also, recall that Algorithm~\ref{ALGO:solve} firstly collects a batch $\B^{(\iota_j)}$ of maximal cliques, 
through the execution of lines~4-7. Each maximal clique which is added to $\B^{(\iota_j)}$ at line~6  
would also be printed out at line~7 of Algorithm~\ref{ALGO:solve}. 
However, all of these output printings are hooked at line~9 of Algorithm~\ref{ALGO:strict}. Thus, each   
maximal clique $K$ within $\B^{(\iota_j)}$ is not printed out immediately (\ie at the time of the hooking), 
instead $K$ is added to $Q$ soon after at line~10 of Algorithm~\ref{ALGO:strict}. 
The rest of the analysis is divided in two cases. 
\begin{itemize}
\item \emph{Case~1.} 
If $|\B^{(\iota_j)}|=n^2$, then (as already observed in Fact~3 within the proof of Proposition~\ref{prop:prealgo_timecomplexity}) 
the simulation of the $\iota_j$ iteration of the \texttt{while-loop} at line~2 of Algorithm~\ref{ALGO:solve} 
takes time at most $c_0 n^{2\omega(1,1,1/2)}$ when $n$ is large enough and for some absolute constant $c_0>0$ 
(whose magnitude highly depends on the rectangular matrix multiplication algorithm employed at line~5 of Algorithm~\ref{ALGO:solve-IB}). 
As already observed in Fact~2 within the proof of Proposition~\ref{prop:prealgo_timecomplexity}, 
all maximal cliques in $\B^{(\iota_j)}$ are outputted with $O(n^2)$ time delay. 
These facts imply that Algorithm~\ref{ALGO:strict} can remove one element from $Q$ 
(at line~12) and print it out (soon after at line~13) every $\tau_{\texttt{delay}}=c_0\lceil n^{2\omega(1,1,1/2)-2}\rceil$ steps, 
without ever emptying $Q$ for this (provided $c_0$ is a sufficiently large constant, and provided $n$ is large enough); 
stated otherwise, during the simulation of $\iota_j$, 
at each iteration of line~11 it must hold $|Q|>0$ whenever $\texttt{counter}\geq \tau_{\texttt{delay}}$. 
Thus, Algorithm~\ref{ALGO:strict} actually outputs a maximal clique of $G$ at line~13 every 
$\tau_{\texttt{delay}}=O\big(n^{2\omega(1,1,1/2)-2}\big) = O\big(n^{2.093362}\big)$ steps.

As a side note, this also implies that at the last step of any such $\iota_j$ the queue $Q$ must contain at least 
as many elements as it contained at the first step of $\iota_j$.

Indeed, observe that: \[ \frac{c_0 n^{2\omega(1,1,1/2)}}{c_0 \lceil n^{2\omega(1,1,1/2)-2}\rceil}\leq n^2=|\B^{(\iota_j)}|.\]
\item \emph{Case~2.} If $|\B^{(\iota_j)}|< n^2$, then (as already observed in Fact~3 
within the proof of Proposition~\ref{prop:prealgo_timecomplexity}) the simulation of the 
$\iota_j$-th iteration of the \texttt{while-loop} at line~2 of 
Algorithm~\ref{ALGO:solve} takes time at most $c'_0 n^{\omega}|\B^{(\iota_j)}|< c'_0 n^{\omega+2}$ 
when $n$ is large enough and for some absolute constant $c'_0>0$ 
(whose magnitude highly depends on the square matrix multiplication algorithm that is  
employed at line~4 of \texttt{children()} within Algorithm~\ref{ALGO:main}). 
By Proposition~\ref{prop:prealgo_timecomplexity}, 
this \emph{Case~2} can occur at most $n$ times during the whole simulation of Algorithm~\ref{ALGO:solve}, 
so that the total (aggregate) time complexity that may be consumed (across all such possible occurrences of \emph{Case~2}) is at most 
 $c'_0 n^{\omega}|\B^{(\iota_j)}|n< c'_0 n^{\omega+3}$. 
Now, recall that Algorithm~\ref{ALGO:boot} had collected at least 
$T= c_1 \lceil n^{\omega - 2\omega(1,1,1/2) + 5}\rceil$ maximal cliques inside $Q$, 
for some absolute constant $c_1>0$. 
Also recall that, at each occurrence $\iota_{j'}$ of \emph{Case~1}, the queue $Q$ must contain  at the last step of $\iota_{j'}$
at least as many elements as it contained at the first step of $\iota_{j'}$. 
Finally, let us assume (without loss of generality) that we had picked $c_1$ such that $c_1\geq c'_0/c_0$.
These facts imply that, during any occurence $\iota_j$ of \emph{Case~2}, 
Algorithm~\ref{ALGO:strict} can remove one element from $Q$ (at line~12) 
and print it out (soon after at line~13) every $\tau_{\texttt{delay}}=c_0\lceil n^{2\omega(1,1,1/2)-2}\rceil$ steps, 
without ever emptying the queue $Q$ for this (provided that $c_0,c_1$ are sufficiently large absolute constants, and provided that 
$n$ is large enough); stated otherwise, during the simulation of $\iota_j$, 
at each iteration of line~11 it must hold $|Q|>0$ whenever $\texttt{counter}\geq \tau_{\texttt{delay}}$.
Thus, also in this \emph{Case~2}, Algorithm~\ref{ALGO:strict} actually prints out a maximal clique of $G$ at line~13 every 
$\tau_{\texttt{delay}}=O\big(n^{2\omega(1,1,1/2)-2}\big) = O\big(n^{2.093362}\big)$ steps.

Indeed, observe that: \[\frac{c'_0 n^{\omega+3}}{c_0\lceil 
n^{2\omega(1,1,1/2)-2}\rceil}\leq c_1 n^{\omega - 2\omega(1,1,1/2)+5} \leq T.\]
\end{itemize} 
Since there are no other cases to take into account, this suffices to conclude the proof. 
\end{proof}

\begin{Prop}\label{prop:finalspace} The overall space usage of Algorithm~\ref{ALGO:strict} is 
$O(n^{\omega - 2\omega(1,1,1/2) + 6})=O(n^{4.2795019})$. 
\end{Prop}
\begin{proof} 
Recall that Algorithm~\ref{ALGO:strict} performs a simulation of Algorithm~\ref{ALGO:main}, which 
consumes at most $O(n^4)$ space by Proposition~\ref{prop:prealgo_space}. The queue $Q$ maintained by 
Algorithm~\ref{ALGO:strict} can grow up to contain at most $T+n^2$ elements 
(because of lines~11-14 of Algorithm~\ref{ALGO:strict}). Since 
$T=O(n^{\omega - 2\omega(1,1,1/2) + 5})=O(n^{3.2795019})$ by definition, 
and each maximal clique has size $O(n)$, then the overall space usage of the procedure is 
$O(n^{\omega - 2\omega(1,1,1/2) + 6})=O(n^{4.2795019})$.
\end{proof}

At this point, Theorem~\ref{thm:strict} follows from Proposition~\ref{prop:finalcorrect}, Proposition~\ref{prop:finaldelay} and Proposition~\ref{prop:finalspace}.

%% file: appendix-time_complexity.tex
In this section it is shown that, according to the current upper bounds on rectangular fast matrix multiplication \cite{Gall12, Huang98}, 
the optimal size of the batch of maximal cliques $\B$ turns out to be $|\B|=n^2$. To start with, recall from Proposition~\ref{prop:reduction} that $M_{\B,G}$ 
can be computed by performing an $|\B|\times n$ by $n\times n^2$ matrix product. Also recall that, by computing $M_{\B,G}$, one solves in one single shot $|\B|$ problem instances, 
\ie $\I^P$ for every $P\in\B$. Let $k\in\Q$ be such that $|\B|=\lfloor n^k\rfloor$. 
Then, computing $M_{\B,G}$, the \emph{amortized} time $\texttt{Time}\big[\I^{P}\big]$ for solving each problem $\I^P$ for $P\in B$ can be bounded as follows,  
where $\texttt{Time}\big[M_{\B,G}\big]$ denotes the time it takes to compute the matrix product $M_{\B}\cdot M_{G}$:
\begin{align*}
\texttt{Time}\big[\I^{P}\big] &   = O\Big(\frac{\texttt{Time}\big[\I^{\B}\big]}{|\B|}\Big) 
                                  = O\Big(\frac{\texttt{Time}\big[\K^{\B}\big]+|\B|n^2}{|\B|}\Big) \\ 
			      &   = O\Big(\frac{\texttt{Time}\big[M_{\B,G}\big]}{\lfloor n^k \rfloor} + n^2\Big) \\ 
			      &   = O\Big( n^{\omega(k, 1, 2)-k} + n^2\Big) \\  
\end{align*}
Our aim would be to find $k\in [0, +\infty)\cap\Q$ such that $\omega(k, 1, 2)-k$ attains its global minimum value. 
Even though the exact values of $\omega(k, 1, 2)$ are currently unknown, 
one can nevertheless minimize the functions that arise from the currently known upper bounds on $\omega(k, 1, 2)$.
In this work we consider the bound $f_{\text{HP98}}$ of Huang and Pan~\cite{Huang98}, and the bound $f_{\text{LG12}}$ of Le~Gall~\cite{Gall12}. 
In particular, $f_{\text{LG12}}$ gives the best upper bound on $\omega(2, 1, 2)=2\omega(1,1,1/2)$ which is currently known in the literature. 
These bounds have been obtained within the framework of \emph{bilinear algorithms}~\cite{Huang98, Gall12}.
Indeed, presently and historically, all the known algorithms supporting the record asymptotic complexity 
estimates for matrix multiplication have been devised as bilinear algorithms. In such framework, 
the complexity bounds are expressed in terms of the minimum number of bilinear multiplications needed for the computation, 
as it can be shown that the number of arithmetic additions or scalar multiplications affect the cost only in a negligible way, see e.g.~\cite{Huang98, Gall12}.
Stated otherwise, it is known in the literature~\cite{Huang98} that the minimum number $R(m,n,p)$ of bilinear multiplications 
used in all bilinear algorithms for $m\times n$ by $n\times p$ matrix multiplications 
is an appropriate measure for the corresponding (arithmetic) asymptotic complexity $C(m,n,p)$.

The following equalities are also known in the literature, see e.g.~\cite{Huang98}.
\begin{align*}
\omega(ar, as, at) &= a\omega(r,s,t) \;\;\;\;\;\;\;\;\;\;\;\;\;\; \text{(\emph{homogeneity})} \\
\omega(r,s,t) = \omega(r,t,s) = \omega(s, r, t) &= \omega(s, t, r) = \omega(t, r, s) = \omega(t, s, r)
\end{align*}
Finally, we shall express our bounds by considering the following two quantities:
\begin{align*}
\alpha & =\sup\{k\in \Q \mid \omega(1,1,k)=2\} >  0.30298 &  \text{(ref.~\cite{Gall12})}\\
\omega & = \omega(1,1,1)  \leq  2.3728639 & \text{(ref.~\cite{Gall14})}
\end{align*}

\subsection{The upper bounds of Huang and Pan (1998).}

The first function $g_1$ is defined as follows: \[g_1(k)= f_{\text{HP98}}(k)-k,\text{ for every } k\in[0, +\infty),\] 
here, $f_{\text{HP98}}:[0, +\infty)\rightarrow\RR$ 
is a piecewise-linear function, which was essentially pointed out by Huang and Pan in~\cite{Huang98}, 
it satisfies $\omega(k,1,2)\leq f_{\text{HP98}}(k)$ for every $k\in [0, +\infty)\cap\Q$.

Here below, we derive an analytic closed-form formula for $f_{\text{HP98}}(k)$.
\begin{itemize}

\item if $k\in [0, 1)\cap\Q$, we consider the complexity exponent: $\omega(k,1,2)=\omega(2,1,k)$.

The corresponding upper bound is provided in ``\cite{Huang98}, section 8.3, equation 8.2", where $\epsilon>0$ is some small absolute constant: 
\[
f_{\text{HP98}}(1) = \left\{ 
\begin{array}{ll}
 3+\epsilon & \;\;\;\; \text{, if $k\in [0, \alpha]$} \\ 
 \frac{2(1-\alpha) + (1-k) + (\omega-1)(k-\alpha)}{1-\alpha} + \epsilon & \;\;\;\; \text{, if $k\in (\alpha, 1)$}
\end{array}
\right.
\]

\item if $k=1$, we consider the complexity exponent $\omega(1,1,2)$.

The corresponding upper bound is provided in ``\cite{Huang98}, section 8.1, at line~13":
\[
f_{\text{HP98}}(1)    = 3.334 
\]

\item if $k\in (1, 2)\cap\Q$, we consider the complexity exponent: 
\begin{align*}
\omega(k, 1, 2) & = \omega(2,k,1) \\
                & = k\,\omega(2/k, 1, 1/k) 
\end{align*} 

The corresponding upper bound is provided in ``\cite{Huang98}, section 8.3, equation (8.2)", where $\epsilon>0$ is some small absolute constant:
\begin{align*}
f_{\text{HP98}}(k) & = k\cdot\left(\frac{\frac{2}{k}(1-\alpha) + (1-\frac{1}{k}) + (\omega-1)(\frac{1}{k}-\alpha)}{1-\alpha}+\epsilon\right), & 
		\text{\parbox{3.5cm}{for every $k\in (1,2)$.}}
\end{align*}

\item if $k=2$, we consider the complexity exponent: \[\omega(2,1,2)=2\omega(1,1,1/2).\] 
The corresponding upper bound is provided in ``\cite{Huang98}, section 8.2, equation (8.1)":
\begin{align*}
f_{\text{HP98}}(2) & = 2\cdot \left.\frac{2(1-r) + (r-\alpha)\omega}{1-\alpha}\right|_{r=\frac{1}{2}} \\
     & = 4.2107878		
\end{align*}

\item if $k\in (2,+\infty)\cap\Q$, we consider the complexity exponent: 
\begin{align*}
	\omega(k,1,2) & = \omega(k,2,1) \\ 
	              & = 2\,\omega(k/2, 1, 1/2)      
\end{align*}
The corresponding upper bound is provided in ``\cite{Huang98}, section 8.2, equation (8.1)":
\begin{align*}
f_{\text{HP98}}(k) & = 2\cdot\left(\frac{\frac{k}{2}(1-\alpha) + \frac{1}{2} 
	+ (\omega-1)(\frac{1}{2}-\alpha) }{1-\alpha} + \epsilon \right), \;\;\;\;
 		\text{\parbox{3.5cm}{for every $k\in (2,+\infty)$ and some small $\epsilon>0$.}}
\end{align*}

\end{itemize}

In summary, $g_1$ can be defined by the following formula, for every $k\in [0, +\infty)$ and $\epsilon>0$ is some small absolute constant:
\[
g_1(k) = f_{\text{HP98}}(k) - k = \left\{
\begin{array}{ll}
 3 - k + \epsilon  & \;\;\;\; \text{, if $k\in [0, \alpha]$} \\ 
 \frac{2(1-\alpha) + (1-k) + (\omega-1)(k-\alpha)}{1-\alpha} - k + \epsilon & \;\;\;\; \text{, if $k\in (\alpha, 1)$} \\ 
2.334 & \;\;\;\; \text{, if $k=1$} \\ 
k\cdot\left(\frac{\frac{2}{k}(1-\alpha) + (1-\frac{1}{k}) + (\omega-1)(\frac{1}{k}-\alpha)}{1-\alpha}+\epsilon\right) - k & \;\;\;\; \text{, if $k\in (1,2)$.} \\ 
2.2107878 & \;\;\;\; \text{, if $k=2$} \\ 
2\cdot\left(\frac{\frac{k}{2}(1-\alpha) + \frac{1}{2} 
	+ (\omega-1)(\frac{1}{2}-\alpha) }{1-\alpha} + \epsilon \right) - k & \;\;\;\; \text{, if $k\in (2, +\infty)$} 
\end{array}
\right.
\]

The qualitative behaviour of $g_1$ is traced in \figref{fig:plot_bound1}, with a filled blue colored line. 

At this point, observe that $g_1(k)$ is piecewise linear and that it attains its global minimum for $k = 2$, 
\ie \[ g_1(2)= \min_{k\in [0, +\infty)} g_1(k) = 2.2107878 .\] 


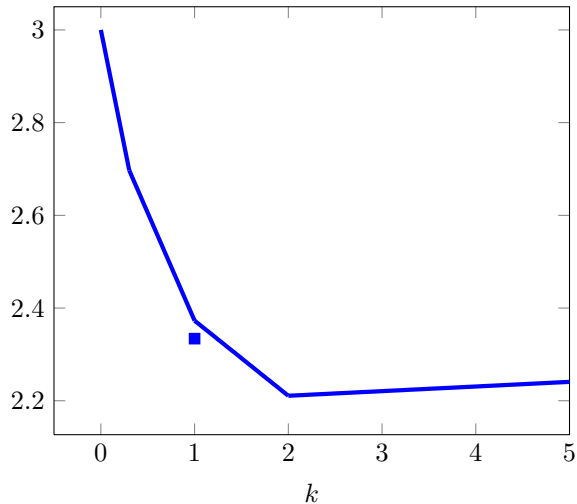
\begin{figure}[h!] 
\centering
\begin{tikzpicture}
\begin{axis}[xmax=5, ymax=3.05, xlabel={$k$}]
\newcommand\alphafm{0.30298}
\newcommand\omegafm{2.3728639}
\newcommand\messfm{(2*(1-\alphafm)  + (\omegafm -1)*(1-\alphafm))/(1-\alphafm)-1}
\newcommand\bigmessfm{( 2*(1-\alphafm) + (\omegafm-1)*(1-\alphafm) ) / (1-\alphafm) + 0.05 - 1}

  \addplot[
	blue, ultra thick,   
	domain= 0 : \alphafm,
	samples=50
  ] { 3-x };

  \addplot[
	blue, ultra thick,  
	domain= \alphafm : 1,
	samples=50
  ] { (2*(1-\alphafm) + (1-x) + (\omegafm -1)*(x-\alphafm))/(1-\alphafm)-x };

  \addplot[
	scatter,
	only marks, 
	point meta=explicit symbolic, 
	scatter/classes={
		a={mark=square*, blue} %
	},
   ] table[meta=label] {
	x   y          label 
  	1   2.334   a  
  };

  \addplot[
	blue, ultra thick, 
	domain= 1 : 2,
	samples=50
  ] { x*( ( (2/x)*(1-\alphafm) + (1-(1/x)) + (\omegafm-1)*((1/x)-\alphafm) ) / (1-\alphafm) ) - x };

  \addplot[
	blue, ultra thick,  
	domain= 2 : 5,
	samples=50
  ] { 2*( (x/2)*(1-\alphafm) + 0.5 + (\omegafm-1)*(0.5-\alphafm)) / (1-\alphafm) - 0.99*x - 0.02 };




\end{axis} 
\end{tikzpicture} 
\caption{Plot of $g_1(k)$ for $k\in [0,5]$.}\label{fig:plot_bound1} 
\end{figure} 

\subsection{The upper bounds of Le~Gall (2012)}

In a similar way, the second function $g_2$ is defined as follows: 
\[ g_2(k) = k f_{\text{LG12}}(k)-k, \text{ for every } k\in \{1\}\cup [2, +\infty). \] 
here, $g_2$ takes into account the upper bound $f_{\text{LG12}}(k)$ for 
$\omega(1,1,1/k)$, which was established by Le~Gall~in~\cite{Gall12}.
We remark that, at the current state of the art, the upper bounds of Le~Gall apply to $\omega(r,s,t)$ if and only if $r=s$.
For this reason, when $k\in (2,+\infty)\cap\Q$, we were able to apply Le~Gall's bounds on 
$\omega(k,1,2)$ only by relying on the following upper bound: 
\begin{align*}
\omega(k,1,2) & \leq  \omega(k,1,k) & \text{(for every  $k\geq 2$)} \\ 
              &  =   k\,\omega(1,1,1/k) & \text{(by homogeneity)} \\
	      & \leq k f_{\text{LG12}}(k). &  
\end{align*}
\begin{figure}[h!]
\centering
	\begin{tabular}[b]{| r | r |}
		\hline
	 	\multicolumn{1}{|c|}{$k$}  & \multicolumn{1}{c|}{$g_2$} \\
		\hline
		  1 & 2.256689 \\ 
                  $0.5^{-1}$ & 2.093362 \\
		  $0.45^{-1}$ &  2.2824489 \\
                  $0.4^{-1}$ &  2.5304375 \\ 
                  $0.35^{-1}$ &  2.8653429 \\ 
                  $0.34^{-1}$ &  2.9463853 \\ 
                  $0.33^{-1}$ &  3.0331485 \\  
		\hline
	\end{tabular}
\caption{Some values of $g_2$}\label{tab:LeGallData}	
\end{figure}
In addition, we applied Le~Gall's bounds on $n\times n$ by $n\times n^2$ matrix products by considering the complexity 
exponent $\omega(1,1,2)$, which is actually one of those explicitly studied by Le~Gall~in~\cite{Gall12}. 
Notice that, when $k\in (1,2)$, it is not possible to apply the results of Le~Gall~\cite{Gall12} 
to bound $\omega(k,1,2)$, because $k\neq 1, k\neq 2$ and $1\neq 2$ so that the above mentioned 
condition (\ie that $r=s$ in $\omega(r,s,t)$) doesn't apply in that case. 
This explains why $g_2(k)$ is defined on $k\in \{1\}\cup [2, +\infty)$.
The qualitative behaviour, and many exact values, of $f_{\text{LG12}}$ were evaluated in \cite{Gall12}, 
by solving a non-linear optimization problem with the computer program Maple (see ``\cite{Gall12}, Section~1, page~4, Table~1 and Figure~1"). 
Here above, in \figref{tab:LeGallData}, we provide some data for $g_2(k)$.
This allows us to show the qualitative behaviour of $g_2$, 
as it is traced in \figref{fig:plot_bound} with a dashed red colored line.
In summary, the results in~\cite{Gall12} allow us to assert that 
$g_2(1)=2.256689 > 2.093362=g_2(2)$ and that $g_2(k)$ is monotone increasing in $[2, +\infty)$.
Concerning its global minimization, since $g_2(1)=2.256689 > 2.093362=g_2(2)$ 
and since $g_2(k)$ is monotone increasing in $[2, +\infty)$, the following holds:
\[
g_2(2) = \min_{k\in \{1\}\cup[2, +\infty)} g_2(k) = 2.093362. 
\]

The qualitative behaviour of $g_2$ is traced in \figref{fig:plot_bound1}, with a dashed red colored line. 
The graphic shows that $g_1$ and $g_2$ perfectly agree on their argument of minimum value, which is $k=2$.

\begin{figure}[h!] 
\centering
\begin{tikzpicture}
\begin{axis}[xmax=5, ymax=3.5, xlabel={$k$}, 
legend entries={\text{$g_1(k)$, ref.\cite{Huang98}},
                \text{$g_2(k)$, ref.\cite{Gall12}}}, 
legend pos=outer north east
]
\addlegendimage{only marks,blue}
\addlegendimage{only marks,red}

\newcommand\alphafm{0.30298}
\newcommand\omegafm{2.3728639}
\newcommand\messfm{(2*(1-\alphafm)  + (\omegafm -1)*(1-\alphafm))/(1-\alphafm)-1}
\newcommand\bigmessfm{( 2*(1-\alphafm) + (\omegafm-1)*(1-\alphafm) ) / (1-\alphafm) + 0.05 - 1}

  \addplot[
	blue, ultra thick,   
	domain= 0 : \alphafm,
	samples=50
  ] { 3-x };

  \addplot[
	blue, ultra thick,  
	domain= \alphafm : 1,
	samples=50
  ] { (2*(1-\alphafm) + (1-x) + (\omegafm -1)*(x-\alphafm))/(1-\alphafm)-x };

  \addplot[
	scatter,
	only marks, 
	point meta=explicit symbolic, 
	scatter/classes={
		a={mark=square*, blue} %
	},
   ] table[meta=label] {
	x   y          label 
  	1   2.334   a  
  };

  \addplot[
	blue, ultra thick, 
	domain= 1 : 2,
	samples=50
  ] { x*( ( (2/x)*(1-\alphafm) + (1-(1/x)) + (\omegafm-1)*((1/x)-\alphafm) ) / (1-\alphafm) ) - x };

  \addplot[
	blue, ultra thick,  
	domain= 2 : 5,
	samples=50
  ] { 2*( (x/2)*(1-\alphafm) + 0.5 + (\omegafm-1)*(0.5-\alphafm)) / (1-\alphafm) - 0.99*x - 0.02 };

  \addplot[
	scatter,
	only marks, 
	point meta=explicit symbolic, 
	scatter/classes={
		a={mark=square*, red} %
	},
   ] table[meta=label] {
	x   y          label 
  	1   2.256689   a  
  };

  \addplot[red, dotted] coordinates {
  	(1, 2.256689) (1/0.5, 2.093362)   
  };

  \addplot[red, ultra thick, dashed] coordinates {
  	(1/0.5, 2.093362) (1/0.45, 2.2824489) (1/0.4, 2.5304375) (1/0.35, 2.8653429) (1/0.34, 2.9463853)
        (1/0.33, 3.0331485) (1/0.32, 3.1261594) (1/0.31, 3.2260097) (1/0.30298, 3.3005479)
  };

  \addplot[
	red, ultra thick, dashed,   
	domain= 1/0.30298 : 5,
	samples=50
  ] { x*2-x };

\end{axis} 
\end{tikzpicture} 
\caption{Plot of $g_1(k)$ for $k\in [0,5]$, and of $g_2(k)$ for $k\in\{1\}\cup[2, +\infty)$.}\label{fig:plot_bound1} 
\end{figure}
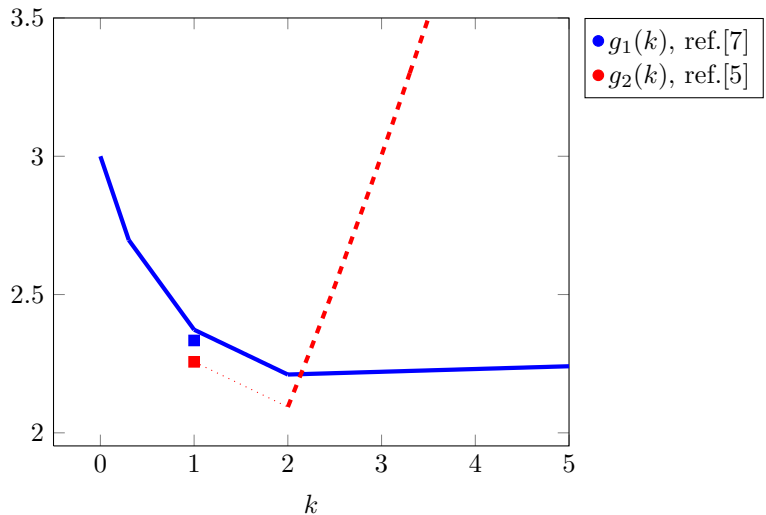 

\paragraph{Conclusion of Appendix A.}
In summary, both estimates~\cite{Gall12}~and~\cite{Huang98} indicate that the minimum complexity comes at $k=2$, 
namely, they both indicate that the optimal size of the batch of maximal cliques $\B$ is given by $|\B|=n^2$.

%% file: appendix-bitwise_and.tex
In this section it is shown how to reduce $\I^{P}$ (\ie the problem of generating all the maximal clique children $C$ of 
any maximal clique parent $P$) to the following query problem named {\sc QSFI}.

\smallskip
\framebox{
\begin{minipage}{12cm}
\textsc{Problem:} {\sc Query-Set-Family-Intersection (QSFI)}.

\rule{\textwidth}{1pt}
They are fixed two positive integers $m,n\in\N$ 
and a family of sets $\{S_k\}_{k\in [m]}$ such that $S_k\subseteq [n]$ for every $k\in [m]$.

\textsc{Task:} Fixed $\langle m, n, \{S_k\}_{k\in [m]}\rangle$, 
decide (possibly many) queries $Q(P,k)$ of the following form: 
\[Q(P,k)=\text{``is it true that $P\cap S_k\neq\emptyset$ holds ?"}, \]
where $P\subseteq [n]$ and $k\in [m]$.
\end{minipage}
}
\smallskip

The reduction to {\sc QSFI} follows here below.
\begin{Prop}[Reduction from $\I^{P}$ to {\sc QSFI}]\label{prop:reduction2SFI}
Let $P$ be a maximal clique of any given $n$-vertex graph $G=(V,E)$. 
For every $i,j\in [n]$, define the following sets $A_i,B_j\subseteq [n]$: 
\[A_i=V_{< i}\cap \Gamma(i) \text{ and } B_j=\Gamma(j).\] 
Moreover, let $m=n^2$ and consider the set family $\{S_{i,j}\}_{i,j\in [n]}$ which is defined as follows:
\[
	S_{i,j} = A_i\setminus B_j \text{, for every $i,j\in [n]$}. 
\]
Now, consider an instance of QSFI in which $\langle m, n, \{S_{i,j}\}_{i,j\in [n]}\rangle$ is fixed as we have just mentioned. 
Also recall that $(i,j)$ is \emph{good} with respect to $P$ 
if and only if there exists $u\in P_{< i}\cap \Gamma(i)$ such that $\{u, j\}\not\in E$.
Then, the following holds:
\[
(i,j)\in [n]\times [n] \text{ is good w.r.t. } P \iff Q(P, (i,j)) \text{ is } \texttt{YES} \iff P\cap S_{i,j}\neq\emptyset.
\]
\end{Prop}
\begin{proof}
The proof is similar to that of Proposition~\ref{prop:reduction}.
To start with, observe that $j\in V$ is adjacent to all the vertices in $P_{< i}\cap \Gamma(i)$ if and only if 
$\big( P_{< i}\cap \Gamma(i)\big) \setminus \Gamma(j)=\emptyset$. Equivalently, \[(i,j)\in [n]\times[n] \text{ is good \wrt $P$  } \iff
 P \cap\Big(\big( V_{< i}\cap \Gamma(i)\big) \setminus \Gamma(j)\Big)\neq\emptyset.\] 
Clearly, $\big( V_{< i}\cap \Gamma(i)\big) \setminus \Gamma(j)$ 
depends only on $i,j$ and not on $P$, so that one can safely write this set as $S_{i,j}=A_i\setminus B_j=\big( V_{< i}\cap \Gamma(i)\big) \setminus \Gamma(j)$. 
Thus, in order to assess whether $(i,j)$ is good with respect to $P$, 
it is sufficient to check whether or not $P\cap S_{i,j}\neq \emptyset$ holds.
\end{proof}

In practice, in order to test the $Q(P,k)$ queries, one could advantageously exploit some fast set intersection algorithms based on Bitwise-AND and SIMD instructions. 
However, it currently remains an open question to determine how these techniques compare in practice with some other well known algorithms for {\mainproblem} such as the Bron-Kerbosch and derived algorithms.